\documentclass[11pt]{article}

\usepackage{bm, pgf, tikz}
\usetikzlibrary{arrows, automata}
\usepackage{subcaption, authblk, lineno} 
\RequirePackage{amsthm, amsmath, amsfonts, amssymb}
\RequirePackage[round, authoryear]{natbib}
\usepackage{url}
\usepackage[colorlinks,linkcolor={blue},urlcolor={blue}, citecolor={blue}, filecolor={black}]{hyperref}
\RequirePackage{graphicx, color, multirow, caption, booktabs, float, mathrsfs, enumitem}
\usepackage[top=1in, bottom=1in, left=0.75in, right=0.75in]{geometry}
\usepackage[width=.85\textwidth]{caption}
\usepackage[figuresright]{rotating} 
\usepackage{adjustbox}
\numberwithin{equation}{section}
\theoremstyle{plain}
\newtheorem{theorem}{Theorem}[section]
\newtheorem{lemma}{Lemma}

\theoremstyle{remark}

\newlist{assumpenum}{enumerate}{1}
\setlist[assumpenum]{leftmargin=3.5em, label={\bf (A.\arabic*)}, ref=(A.\arabic*)}
\newlist{assumpenumA2}{enumerate}{1}
\setlist[assumpenumA2]{leftmargin=3em, label={\bf (A.2.\arabic*)}, ref=(A.2.\arabic*)}

\newlist{estenum}{enumerate}{1}
\setlist[estenum]{leftmargin=3.5em, label={\bf \arabic*)}, ref=\arabic*)}

\newlist{lemmaitemenum}{enumerate}{1}
\setlist[lemmaitemenum]{leftmargin=4em, label=(\thelemma.\arabic*), ref=(\thelemma.\arabic*)}

\def\lcrarrow#1{\overset{#1}{\longrightarrow}}
\def\bs#1{\boldsymbol{#1}}
\def\Pn{\mathbb{P}_n}
\newcommand\independent{\protect\mathpalette{\protect\independenT}{\perp}}
\def\independenT#1#2{\mathrel{\rlap{$#1#2$}\mkern2mu{#1#2}}}

\definecolor{change}{RGB}{0,0,0}

\title{Post-selection inference for high-dimensional
mediation analysis with survival outcomes}

\author[1]{Tzu-Jung Huang}
\author[2]{Zhonghua Liu}
\author[2]{Ian W. McKeague}

\affil[1]{Vaccine and Infectious Disease Division, Fred Hutchinson Cancer Research Center, WA U.S.A}
\affil[2]{Department of Biostatistics, Columbia University, NY, U.S.A}

\date{}

\begin{document}
\allowdisplaybreaks
\maketitle

\abstract{It is of substantial scientific interest to detect mediators that lie in the causal pathway from an exposure to a survival outcome. However,  with high-dimensional mediators, as often encountered in modern genomic data settings, there is a lack of powerful methods that can provide valid post-selection inference for the identified marginal mediation effect. To resolve this challenge, we develop a post-selection inference procedure for the maximally selected natural indirect effect using a semiparametric efficient influence function approach. To this end, we establish the asymptotic normality of a stabilized one-step estimator that takes the selection of the mediator into account. Simulation studies show that our proposed method has good empirical performance. We further apply our proposed approach to a lung cancer dataset and find multiple DNA methylation CpG sites that might mediate the effect of cigarette smoking on lung cancer survival.}


\section{Introduction}

Mediation analysis aims to assess whether the effect of an exposure (e.g., smoking) on an outcome of interest (e.g., lung cancer survival) is mediated by an intermediate variable (mediator), for example, DNA methylation \cite{VanderWeele2011,tian2022coxmkf, liu2022large}.
Modern high-throughput platforms typically measure DNA methylation levels at hundreds of thousands of CpG sites, so multiple testing becomes a serious issue that needs to be addressed in a way that scales with the dimension of mediators. Locating CpG sites that mediate the effect of smoking on lung cancer survival offers a way for intervention, as DNA methylation is a reversible process \cite{wu2014reversing}.  

There is a comprehensive literature on causal inference for discrete survival outcomes, e.g.,\ see \cite{vanderweele2015explanation} and Chapter 17 of the monograph of \cite{hernan2023causal}; see also Chapters 18 and 23 of that monograph discussing post-selection inference for high-dimensional predictors and mediation analysis, respectively.  However,
to the best of our knowledge, there is little work in the literature that can address all these aspects for right-censored survival outcomes.  
We focus on the problem of selecting significant mediators in settings where the number of potential mediators is orders-of-magnitude larger than the sample size, motivated by the lung cancer data to be described in the application Section \ref{sec:lungcancer}. In particular, the high dimensionality of DNA methylation data poses a severe challenge to understanding whether DNA methylation mediates the effect of smoking on survival from lung cancer. Existing methods for mediation analysis with survival outcomes might fail to control the family-wise error rate (FWER) unless the selection of the potential mediator is taken into account. \label{knockoff.FWER} Although knock-off methods have been used to control the false-discovery rate (FDR) in this setting \cite{tian2022coxmkf}, they fail to adequately control the FWER; see Chapter 4 of \cite{Efron_2010}. 

The proposed approach involves constructing a semi-parametrically efficient estimator of the association between an exposure (denoted by $A$) and the survival time outcome $T$, as possibly mediated by one or more epigenetic mediators $B_1, \ldots, B_p$, where $p$ can be three or more orders-of-magnitude larger than the sample size. 
This is then used to build a test statistic for detecting the presence of mediators that is computationally tractable and provides effective FWER control.
For uncensored survival outcomes, the classical Sobel test \cite{sobel86} is applicable to our causal inference problem in conjunction with a Bonferroni correction.
However, the asymptotic distribution of Sobel's test statistic for the presence of even a single mediator is known to be discontinuous at the boundary of the null hypothesis of no mediation \cite{liu2022large}. This ``boundary effect'' has a number of consequences for the calibration of Sobel's test. These include inaccurate Type I error rates and failure of the parametric and nonparametric bootstraps: symptoms of the ``quiet scandal of statistics" identified by \cite{breiman92}. \cite{liu2022large} introduce a more powerful competitor, but their test is not designed to control the FWER. As far as we know, our contribution is the first to provide a resolution to the problem of reliably calibrating a test (that controls the FWER) with high-dimensional mediators, even in the case of uncensored outcomes.

The need to control the FWER can be ameliorated in practice via FDR control. This is not a serious drawback when it is known that there are numerous mediators among those targeted, but if the absence of any mediation effects is a possibility, FDR control is inadequate.
We refer interested readers to the monograph of \cite{claeskens_hjort_2008} for background on the general topic of model selection in parametric settings; their discussion in Chapter 10 of issues arising with the estimation of boundary parameters applies to the natural indirect effect in its product form, as used in Sobel's test. In this connection, we also refer to \cite{mckq2015} for an in-depth discussion of the non-standard asymptotics that arise in post-selection inference, specifically in the marginal screening of high-dimensional predictors (of a continuous uncensored response) using FWER control.

\label{intro.Huang2023} Our proposed approach to a post-selection inference problem for high-dimensional mediation analysis is related to \cite{mckq2015} and to recent work of \cite{Huang2019} and \cite{Huang2023}, where post-selection inference for the marginal effect of multiple predictors on a survival outcome was studied. We will refer extensively to the latter paper in the sequel.  
The main contribution of \cite{Huang2023} is to introduce an efficient estimator of each marginal slope parameter in an ``assumption lean'' accelerated failure time (AFT) model for the survival outcome given a specific predictor and a stabilized test statistic based on the maximally selected slope parameter that ``smooths out'' the effect of selecting the slope parameter.  The smoothing step is reminiscent of bagging and leads to a computationally tractable normal calibration for testing purposes, along with a confidence interval for the maximal association between predictors and the outcome.  Further discussion and references on high-dimensional marginal screening in survival analysis can be found in   \cite{Huang2024}.  
In the present setting, the maximally selected  {\it natural indirect effect} (over the considered mediators) is the target parameter of interest for detecting the presence of {\it some} mediation effect, replacing the maximally selected slope parameter when searching for associations between multiple predictors and the survival outcome.

Historically, causal mediation analysis is based on the principle of {\it counterfactual definiteness}: the ability to posit the existence of observations that are in fact not measured (missing). 
Just because we may not see the outcome $T$ for an untreated subject ($A=0$),  we can conceptually ``go back in time'' and think about what would have happened had that subject been treated ($A=1$). This principle emerged in R.A.\ Fisher's work on randomized experiments in 1925,  and in Jerzy Neyman's work on potential outcomes (published in Polish in 1923), see  \cite{rubin90} for discussion. In a series of ground-breaking papers in the 1970s, Donald Rubin made the conceptual leap of treating $T\,|\,{A=0}$ and $T\,|\,{A=1}$ as two separate variables, one of which is missing. According to  \cite{rubin2018}, ``all causal inference problems are  missing data problems.''
Rubin also points out that this idea had in fact emerged in quantum physics in 1927 in the form of Heisenberg's  {\it uncertainty principle}: it is impossible to measure two conjugate variables (e.g., position and momentum) with precision on the same unit \cite{rubin2019essential}. 
The link to causal inference was not noticed at the time, cf.\  \cite{Helland2021}. 

The paper is organized as follows.  In Section \ref{sec:targetparam}, we first review the notion of the {\it natural indirect effect} using counterfactual reasoning (thus providing the basis of our approach in the context of survival data) and then define the target parameter as the \textit{maximal natural indirect effect}. The proposed confidence interval for the maximal mediation effect is developed starting in Sections \ref{sec:eff.if.NIE} and \ref{sec:one-step estimator.k},  culminating in Section \ref{sec:stabilized.one-step}. This is done under minimal assumptions, in terms of an assumption-lean accelerated failure time model for the time-to-event outcome, and a semi-parametrically efficient estimator for each specific mediator. In Sections \ref{sec:simulation}--\ref{sec:confounding}, we conduct extensive simulation studies to assess the empirical performance of our proposed method.  In Section \ref{sec:lungcancer}, we apply our approach to a lung cancer DNA methylation dataset to demonstrate its practical performance.   Proofs are placed in the Appendix.  R code used in our numerical studies is available on a GitHub repository (\url{ https://github.com/
tzujunghuang/High-dim-mediation-analysis-with-survival-outcomes}).

\section{The Target parameter}
\label{sec:targetparam}

In this section, we introduce the notation used in the sequel and define the target parameter of interest: the maximal natural indirect effect of the exposure indicator $A$ on the survival outcome $T$, as mediated by some components of a  $p$-dimensional mediators $\bs{B}$. 

\subsection{Preliminaries}
\label{sec:prelim}
We consider survival data with independent right censorship. Let $T$ and $C$ denote a (log-transformed) survival time and censoring time, respectively. Suppose we observe $n$ i.i.d.\ copies of $O=(X,\delta, A,\bs{B})\sim P$, where $X=\min\{T, C\}$, $\delta=1(T \leq C)$, an exposure variable $A \in \mathcal{A}:= \{0,1\}$,
and $\bs{B}=(B_k, k= 1, \ldots, p)$ is a $p$-vector of ``candidate mediators''. 
The observations are denoted $O_1,\ldots, O_n$, and their empirical distribution by $\mathbb{P}_n$.
Note that $p=p_n$ can grow with $n$, but we omit the subscript $n$ throughout for notational simplicity unless otherwise stated.
We denote the joint distribution of $(T, A,\bs{B})$ by $Q$ and the survival function of the censoring distribution by $G$. 

We assume throughout that the censoring time $C$ is independent of $(T, A,\bs{B})$. The distribution $P$ belongs to the statistical model $\mathcal{M}$, which is the collection of distributions $P$ parameterized by $(Q, G)$ such that $P$ has a density (with respect to an appropriate dominating measure $\nu$) given by
\begin{align*}
  &\frac{dP}{d\nu}(x,\delta,a,\bs{b}) = \big[q(x\,|\,a,\bs{b})G(C\ge x)\big]^{\delta}\big[Q(T\ge x\,|\,A=a,\bs{B}=\bs{b})g(x)\big]^{1-\delta} q(\bs{b}\,|\,a)q(a),
\end{align*}
where $q$ and $g$ are the densities of $Q$ and $G$ with respect to $\nu$. Let the follow-up period be $\mathcal{T}=(-\infty,\tau]$. The sample space is denoted by $\mathcal{X} = \mathcal{T} \times \{0,1\} \times \mathcal{A} \times \mathbb{R}^{p}$ and the empirical distribution on this space is denoted by $\Pn$. 

Our approach specifies the relationship between the exposure variable $A$, each mediator, and the survival outcome, by a general semiparametric accelerated failure time (AFT) model without making any distributional assumption on the error term. The error term is taken to be uncorrelated with the mediators. Specifically, the marginal AFT model takes the form
\begin{align}
  T = \alpha_k + \beta_kB_k + \gamma_kA + \varepsilon_k, \label{eq:linear_model}
\end{align}
where $\alpha_k \in \mathbb{R}$ is an intercept, $\beta_k \in \mathbb{R}$ is the slope parameter for the effect of $B_k$ on $T$, $\gamma_k \in \mathbb{R}$ is the effect of $A$ on $T$, and $\varepsilon_k$
is a zero-mean error term that is uncorrelated with $(A, B_k)$. The model \eqref{eq:linear_model} 
holds without distributional assumptions (such as independent errors) apart from mild moment conditions; see \cite{Huang2023} for a discussion comparing the merits of the AFT model with those of the Cox proportional hazards model.

\subsection{Natural indirect effect}
\label{sec:NIE}
We now introduce the potential outcomes notation and assumptions to be used in the sequel for each $B_k$. For now, we drop the subscript $k$. Let $T(a,b)$ denote the potential outcome of the survival time had the exposure been set to $A=a$ and the mediator been set to  $B=b$, where we assume that the mediator could potentially be manipulated over a range of values.
Also, let $a \mapsto B(a)$ denote the potential mediator had the exposure been set to  $A=a$. 
Then the {\it natural indirect effect} of $A$ on $T$, as mediated by $B$, is given by
${\rm NIE}= E_P[T(1,B(1)) - T(1,B(0))]$.
Similarly, the {\it natural direct effect} of $A$ on $T$ (not  mediated by $B$) is ${\rm NDE} = E_P[T(1,B(0)) - T(0,B(0))]$.
The {\it total effect} of $A$ on $T$ is the sum of its direct and indirect effects:
\begin{align}
 {\rm TE} = E_P[T(1,B(1))] - E_P[T(0,B(0))] = {\rm NDE} + {\rm NIE}.
\end{align}
To identify the NIE, we make the following standard  assumptions adopted in mediation analysis \cite{Imai2010,vanderweele2015explanation}:
\begin{assumpenum}[series=assumptions]
  \item \textit{Consistency:} $T_{obs} =T(A_{obs},B_{obs})$ and $B_{obs} = B(A_{obs})$ almost surely. \label{assump:Consistency}

  \item \textit{Sequential ignorability:} For\, $a, a' \in \{0,1\}$, \label{assump:Ignorability}
  \begin{assumpenumA2}
    \item $(T(a',b), B(a)) \independent A_{obs}$\,; \label{assump:Ignorability.1}
    \item $T(a',b) \independent B\,\big|\,A_{obs} = a$. \label{assump:Ignorability.2}
  \end{assumpenumA2}
  
   \item \textit{Positivity:} $P(A_{obs}=a)>0$ and $P(B(a)\in {\rm d}u\,|\,A_{obs}=a)$ is bounded as a function of $u$, for $a=0,1$. \label{assump:Positivity}
\end{assumpenum}
The sequential ignorability assumption \cite{Imai2010} is also referred to as no unmeasured confounding assumption \cite{vanderweele2015explanation}, which is usually formulated to be conditional on pre-exposure confounders, but for simplicity, we have only given the unconditional version at this stage. In the sequel we will consider how our approach can be adjusted for measured confounders; in fact, we adjust for various measured confounders in the lung cancer data application in Section \ref{sec:lungcancer}.

Now reintroducing the subscript $k$, under assumptions \ref{assump:Consistency}--\ref{assump:Positivity},
the NIE of $A$ on $T$, as mediated by $B_k$, is given by 
\begin{align}\label{eq:def_parameter_k}
  \Psi_k(P) &:= E_P[T(1,B_k(1)) - T(1,B_k(0))] \nonumber \\
  & = \int E_P[T\,|\,A=1, B_k = u]\big\{P(B_k \in {\rm d}u\,|\,A=1) - P(B_k \in {\rm d}u\,|\,A=0)\big\} \nonumber \\
  & = \beta_{k} \int u\big\{P(B_k \in {\rm d}u\,|\,A=1) - P(B_k \in {\rm d}u\,|\,A=0)\big\} = \beta_{k} \zeta_k,
\end{align}
where $\zeta_k =E_P[B_k\,|\,A=1] - E_P[B_k\,|\,A=0]$.
The second line in the above display uses the marginal AFT model \eqref{eq:linear_model} for the conditional mean of $T$ given $(A, B_k)$ and the fact that $E_P[\varepsilon_k\,|\, A, B_k] = 0$, which holds by Theorem 2 of \cite{Imai2010}.

Note that from \eqref{eq:def_parameter_k}, the difference method \cite{Judd1981} and the product-coefficient method \cite{Baron1986} give the identical expression for the natural indirect effect, as noted in \cite{VanderWeele2011}.  The first term $\beta_k$ in the product in \eqref{eq:def_parameter_k}
represents the effect of the mediator $B_k$ on $T$ (under the marginal AFT model), and in our setting is consistently estimated using inverse-probability-weighting of the observed outcome $\delta X$ by the Kaplan--Meier estimator of the censoring survival function, enabling the use of a standard least squares estimator  $\hat{\beta}_{nk}$ as developed by \cite{Koul1981};  $\hat{\beta}_{nk}$ is called the KSV estimator in the sequel. The second term $\zeta_k$ represents the average effect of exposure $A$ on $B_k$. The empirical mean of this difference, $\hat{\zeta}_{nk}$, is naturally used to estimate $\zeta_k$. This in turn yields a consistent estimator ${\color {change} \Psi_k(\hat{P}_n)}:= \hat{\beta}_{nk}\hat{\zeta}_{nk}$ of $\Psi_k(P)$, where $\hat{P}_n$ denotes a combined estimator for the various features of the underlying distribution $P$ that will be introduced in Section \ref{sec:one-step estimator.k}.
 
The target parameter of interest is defined as the maximal natural indirect effect (in absolute value) of $A$ on $T$ as mediated by each individual component of $\bs{B}$:
\begin{align}\label{eq:def_parameter}
 \Psi(P) := \max_{k = 1,\ldots, p}\left|\Psi_{k}(P)\right|.
\end{align}
In the sequel, we develop an asymptotically valid and computationally tractable confidence interval for $\Psi(P)$.
An obviously simple (but not efficient) estimator of $\Psi(P)$ is to plug-in each estimator {\color {change} $\Psi_k(\hat{P}_n)$}
defined above. We note in passing that it is important to make sure each $B_k$ is pre-standardized (in the sense of a normal score) so the magnitudes of the various contributions $|\,\Psi_k(P)\,|$ to $\Psi(P)$ are comparable.  In the sequel, we assume this is the case (without further comments), for both the simulation studies and the real data application.

\section{Efficient influence function for the natural indirect effect}
\label{sec:eff.if.NIE}
\label{sec:efficient_if_psi_k}
In this section, we derive the efficient influence function of the natural indirect effect in the uncensored case and then extend it to the censored case.

For $a, a' \in \{0,1\}$ and $k \in \{1,\ldots,p\}$, define 
\begin{align} \label{eq:def_parameter.2}
  \eta(a,a',k) = \int E_P[T\,|\,A=a,B_k=u]P(B_k \in {\rm d}u\,|\,A=a').
\end{align}
To use the results of \cite{Tchetgen2012}, we first describe the correspondence between our notation and theirs when there is no censoring or confounding. Our survival outcome $T$ plays the role of $Y$ in their notation, and we use $A$ to denote the treatment assignment, whereas they denote it by $E$. 
For each fixed $k$, $\eta(a,a',k)$ here is their $\eta(e,e^*,X)$ with the confounders $X$ removed and $M := B_k$ for $e, e^* \in \{0,1\}$;
moreover, $\eta(1,1,k)$ is equal to their  $\delta_1$. 
Along these lines, their $\theta_0$ is $\eta(1,0,k)$ defined above.
Thus for each fixed $k$, the efficient influence function of $\eta(1,1,k)$ is
\begin{align} \label{eq:if.1}
  (a,t) \mapsto \frac{1(a=1)}{P(A=1)}\Big\{t - \eta(1,1,k)\Big\},
\end{align}
which agrees with their $S_{\delta_1}^{\ {\rm eff,nonpar}}(\delta_1)$ in the absence of any confounding.
In the sequel, we assume that $B_k$ is a discrete random variable for simplicity of notation, although our results naturally apply in the general setting of \cite{Tchetgen2012}.
By their Theorem 1, 
 the efficient influence function of $\eta(1,0,k)$ is given by  their $S_{\theta_0}^{\ {\rm eff,nonpar}}(\theta_0)$ and in our notation is
\begin{align} \label{eq:if.2}
 (a,t,b_k) \mapsto \; &\frac{1(a=1)}{P(A=1)}\frac{P(B_k = b_k\,|\,A=0)}{P(B_k = b_k\,|\,A=1)}\Big\{t - E_P[T\,|\,A=1, B_k = b_k]\Big\} \nonumber \\
 & + \frac{1(a=0)}{P(A=0)}\Big\{E_P[T\,|\,A=1,B_k = b_k] - \eta(1,0,k)\Big\}.
\end{align}
Thus, since the influence function of a difference coincides with the difference of the influence functions, together with expression  \eqref{eq:bayes_rule} and $\Psi_k(P) = \eta(1,1,k) - \eta(1,0,k)$ from the first line of \eqref{eq:def_parameter_k},
the efficient influence function of $\Psi_k(P)$ when $T$ is uncensored is given by
\begin{align} \label{eq:if_psi_k.1}
  & (a,t,b_k) \mapsto 
  - \frac{1(a=0)}{P(A=0)}\Big\{E_P[T\,|\,A=1,B_k = b_k] - \eta(1,0,k)\Big\} \nonumber \\
  & \quad + \frac{1(a=1)}{P(A=1)}\Big\{t - \eta(1,1,k) -
  \frac{P(A=0\,|\,B_k = b_k)P(A=1)}{P(A=1\,|\,B_k = b_k)P(A=0)}\big\{t - E_P[T\,|\,A=1, B_k = b_k]\big\}\Big\},
\end{align}
where we have used the Bayes rule giving
\begin{align} \label{eq:bayes_rule}
  \frac{P(B_k = b_k\,|\,A=0)}{P(B_k = b_k\,|\,A=1)} = \frac{P(A=0\,|\,B_k = b_k)P(A=1)}{P(A=1\,|\,B_k = b_k)P(A=0)}.
\end{align}

In the case that the censoring distribution $G$ is known, the synthetic response $\tilde{Y} = \delta X/G(X)$ has the same first (conditional) moment as $T$, under the assumption of independent censoring:
\begin{align}\label{eq:identical.first.moment}
  E_P[\tilde{Y}\,|\,A, B_k] & = E_P\bigg\{E_P\bigg[\frac{\delta X}{G(X)}\bigg|\,A,B_k,T\bigg]\bigg\} = 
  E_P\bigg\{\frac{T}{G(T)}E_P\bigg[1(T \le C) \bigg|\,A,B_k,T\bigg]\bigg\} \nonumber \\ 
  & = E_P[T\,|\,A,B_k],
\end{align}
for each $k$.
Therefore, the efficient influence function in \eqref{eq:if_psi_k.1} can be re-expressed as, with $o = (a,\delta,x,b_k)$,
\begin{align*}
  & f_k(\cdot\,|\,P) : o \mapsto
   - \frac{1(a=0)}{P(A=0)}\Big\{E_P[\tilde{Y}\,|\,A=1,B_k = b_k] - \eta(1,0,k)\Big\} \nonumber \\
  & \quad + \frac{1(a=1)}{P(A=1)}\Big\{\tilde{y} - \eta(1,1,k) -\frac{P(A=0\,|\,B_k = b_k)P(A=1)}{P(A=1\,|\,B_k = b_k)P(A=0)}\big\{\tilde{y} - E_P[\tilde{Y}\,|\,A=1, B_k = b_k]\big\}\Big\}.
\end{align*}
As we see from \eqref{eq:def_parameter_k} that
$\eta(1,a,k) = \beta_kE_P[B_k\,|\,A=a]$ for $a \in \{0,1\}$, where
$\beta_k$ is the effect of $B_k$ on $T$ in Model \eqref{eq:linear_model}, so under this framework, it is reasonable to re-express the above display as
\begin{align} \label{eq:if_psi_k}
  & f_k(\cdot\,|\,P) : o \mapsto
  - \frac{1(a=0)}{P(A=0)}\Big\{E_P[\tilde{Y}\,|\,A=1,B_k = b_k] - \beta_kE_P[B_k\,|\,A=0]\Big\} + \frac{1(a=1)}{P(A=1)}\Big\{\tilde{y} \nonumber \\
  & - \beta_kE_P[B_k\,|\,A=1] -\frac{P(A=0\,|\,B_k = b_k)P(A=1)}{P(A=1\,|\,B_k = b_k)P(A=0)}\big\{\tilde{y} - E_P[\tilde{Y}\,|\,A=1, B_k = b_k]\big\}\Big\}.
\end{align}

When the censoring distribution $G$ is unknown, however, $f_k(\cdot\,|\, P)$ is no longer the efficient influence function. To obtain the efficient influence function of $\Psi_k(P)$ in this case, we need to project $f_{k}(\cdot\,|\,P)$ onto the tangent space $\mathbf{T}^{\mathcal{M}}(P)$ at $P$ in the model $\mathcal{M}$.
To this end, despite our assumption of independent censoring, it is convenient to consider the broader coarsening-at-random (CAR) model $\mathcal{M}^{car}\supseteq \mathcal{M}$, as indicated in \cite{Huang2023}. Under $\mathcal{M}^{car}$, $G$ is viewed as a survival function for $C$ conditionally on $(A, \bs{B})$, and this survival function may depend non-trivially on $(A, \bs{B})$. Since we have assumed that $C$ is independent of $(T, A,\bs{B})$ for the particular distribution that generates our data, this conditional survival function is equal to the marginal survival function $G(\cdot)$ for that distribution. This observation simplifies the expression for the tangent space for $G$ in $\mathcal{M}^{car}$, which is given by
\begin{align} \label{eq:T_car_space}
  \mathbf{T}^{car}(G) = \left\{\int_{\mathcal{T}} H(A, \bs{B},s)\,M({\rm d}s)\,\middle|\, H \colon \{0,1\} \times \mathbb{R}^{p} \times \mathcal{T} \rightarrow \mathbb{R} \right\},
\end{align}
where $H$ is any measurable function for which the integral has finite variance, and $M({\rm d}s) = 1(X \in {\rm d}s, \delta=0)-1(X \geq s)\, \Lambda({\rm d}s)$ with $\Lambda( \cdot )$ as the cumulative hazard function corresponding to $G(\cdot)$ with respect to the filtration 
\begin{align}\label{eq:filtration}
  \mathcal{F}_{s} = \sigma\big\{1(X \leq s'\,, \delta=0)\,, 1(X \geq s')\,, A\,, \bs{B}, \; s' \le s \in \mathcal{T}\big\}.
\end{align}
See Example 1.12 in \cite{vanderLaan&Robins2003} and Section 3 in \cite{vanderLaan1999} for further details.

Following the techniques given in Appendix A of \cite{Huang2023}, it is shown that $\mathbf{T}^{car}(G)^{\perp}\subseteq \mathbf{T}^{\mathcal{M}}(P)$, where $\perp$ denotes the orthogonal complement in the Hilbert space of $P$-square integrable functions with mean zero, denoted by $L_0^2(P)$.
Using $\Pi(\cdot\,|\,S)$ to denote the projection operator onto a closed linear subspace $S\subseteq L_0^2(P)$, it is also shown that the efficient influence function of $\Psi_k(P)$ should be
$\Pi(\,f_k(\cdot\,|\,P)\,|\,\mathbf{T}^{car}(G)^{\perp})$.
Taking $H(A,\bs{B},s) = -\,E_P[\,f_k(O\,|\,P)\,|\,A,\bs{B},X \geq s]$,
the projection of $f_k(\cdot\,|\,P)$ onto $\mathbf{T}^{car}(G)$ is developed as
\begin{align*} 
  &(a,\delta,x,b_k) \mapsto \;-\int_{\mathcal{T}} E_P\left[\,f_k(O\,|\,P)\,\big|\,A=a,B_k = b_k,X \geq s\right]\big\{1(x \in {\rm d}s, \delta=0)-1(x \geq s)\ \Lambda({\rm d}s)\big\}\\
  & \quad = -\frac{1(a=1)}{P(A=1)}\bigg\{1 -
  \frac{P(A=0\,|\,B_k=b_k)P(A=1)}{P(A=1\,|\,B_k=b_k)P(A=0)}\bigg\}\\
  & \hspace{3cm} \times \int_{\mathcal{T}} 
  E_P\big[\tilde{Y}\,\big|\,A = a,B_k = b_k,X \geq s\big]\big\{1(x \in {\rm d}s, \delta=0)-1(x \geq s)\ \Lambda({\rm d}s)\big\}.
\end{align*}
Hence the efficient influence function of $\Psi_k(P)$, the projection of $f_k(\cdot\,|\,P)$ onto $\mathbf{T}^{car}(G)^{\perp}$, is given by
\begin{align}\label{eq:if_star}
  f^*_k(\cdot\,|\,P) := f_k(\cdot\,|\,P) - f^{car}_k(\cdot\,|\,P),
\end{align}
where $f_k(\cdot\,|\,P)$ is as defined in \eqref{eq:if_psi_k}, $o = (a,\delta,x,b_k)$ and
\begin{align} \label{eq:if_nu_eff}
  f^{car}_k(\cdot\,|\,P) : o \mapsto & -\frac{1(a=1)}{P(A=1)}\bigg\{1- \frac{P(A=0\,|\,B_k=b_k)P(A=1)}{P(A=1\,|\,B_k=b_k)P(A=0)}\bigg\}\\
  & \quad \times \int_{\mathcal{T}} E_P\big[\tilde{Y}\,\big|\,A = a,B_k = b_k,X \geq s\big]\big\{1(x \in {\rm d}s, \delta=0)-1(x \geq s)\ \Lambda({\rm d}s)\big\}. \nonumber 
\end{align}

\section{One-step estimator of the natural indirect effect}  
\label{sec:one-step estimator.k}
This section utilizes the efficient influence function given above to construct an asymptotically efficient one-step estimator of $\Psi_k(P)$. 
\cite{Huang2023} introduce an efficient estimator of the {\it maximal marginal} effect of high-dimensional predictors on the survival outcome, but estimating the {\it natural indirect effect} mediated by high-dimensional mediators poses a new challenge.

To build the proposed estimator we need to identify the various features of $P$ needed to estimate \eqref{eq:if_psi_k} and \eqref{eq:if_nu_eff}, namely the following six ingredients. These will be combined into an estimator denoted by $\hat{P}_n$:
\begin{estenum}
  \item $\hat{E}_n(a,u,s,k)\colon$ an estimator of $E_{P}[\tilde{Y}\,|\,A=a, B_k=u, X \ge s]$ restricted to take values in some $P$-Donsker family of uniformly bounded functions of $(a,u,s)\in \mathcal{A} \times \mathbb{R} \times \mathcal{T}$. 
  \label{item:estimation_E}

  \item $\hat{G}_{n}(\cdot)\colon$ the Kaplan--Meier estimator of the censoring survival function $G$. Note that $Y = \delta X/\hat{G}_n(X)$  is an inverse-probability-weighted estimator of the synthetic response $\tilde{Y}$ used in the KSV estimator of $\beta_k$ mentioned earlier. This in turn provides an estimator $\hat{\Lambda}_n$ of the cumulative hazard function $\Lambda$ of the censoring distribution, and
  {\color {change} $\hat{M}({\rm d}s) = 1(X \in {\rm d}s, \delta=0)-1(X \geq s)\, \hat{\Lambda}_n({\rm d}s)$.}
  \label{item:estimation_G}

  \item $\mathbb{Q}_{n}(a) \colon$ the empirical estimator of $Q(a) := P(A=a)$; in our binary experiment, $\mathbb{Q}_{n}(1) = \mathbb{P}_{n}[1(A=1)]$ and
  $\mathbb{Q}_{n}(0) = 1-\mathbb{Q}_{n}(1)$. \label{item:estimation_prob}

  \item {\color {change} $\hat{q}_{n}(a,k) \colon$} the empirical estimator of $q(a,k) := E_P[B_k\,|\,A=a]$, the sample mean of $B_k$ among subjects with $A=a$. Note that $\hat{\zeta}_{nk} = \hat{q}_{n}(1,k) - \hat{q}_{n}(0,k)$. \label{item:estimation_conditional.mean}

  \item $\hat{\beta}_{nk} \colon$ the KSV estimator for $\beta_k$, the effect of $B_k$ on $T$ in Model \eqref{eq:linear_model}, needed to estimate \eqref{eq:if_psi_k}. Specifically,
  \begin{align}
     \hat{\beta}_{nk} := \frac{{\rm Var}_{\mathbb{P}_{n}}(A){\rm Cov}_{\mathbb{P}_{n}}(B_k,Y)-{\rm Cov}_{\mathbb{P}_{n}}(A,B_k){\rm Cov}_{\mathbb{P}_{n}}(A,Y)}{{\rm Var}_{\mathbb{P}_n}(A){\rm Var}_{\mathbb{P}_n}(B_k) - {\rm Cov}^2_{\mathbb{P}_{n}}(A,B_k)}
  \end{align}
  in which {\color {change}
  $Y$ is defined in {\bf \ref{item:estimation_G}}.} \label{item:estimation_beta.k}

  \item $\hat{Q}_{n}(u,k)\colon$ 
  a uniformly consistent estimator of the reciprocal of the odds for the counterfactual effect of $B_k$ on $A$, namely $Q(u,k):=P(A=0\,|\, B_k=u)\big/P(A=1\,|\, B_k=u)$, as a function of $u$, for each $k$.
  In practice, it would be reasonable to use logistic regression to estimate $Q(u,k)$, as we do for the sake of simplicity, although from Bayes rule we could in theory estimate it nonparametrically.\label{item:estimation_reciprocal.odds}
\end{estenum}  
Note that the estimators \ref{item:estimation_E} and \ref{item:estimation_G} furnish an estimator of  $E_P[T\,|\,A=a, B_k=u, X \ge s]$, in view of  $$E_{P}[T\,|\,A=a,B_k=u,X \ge s] = G(s)E_{P}[\tilde{Y}\,|\,A=a,B_k=u,X \geq s].$$

For each  $a \in \mathcal{A}$, we introduce  the estimator $\hat{E}_n(a,u,s,k)$ constructed by regressing $Y$ on $B_k$ using only a sub-sample $\{O_i=(X_i,\delta_i,\bs{B}_i)\, ,\, i:A_i=a,X_i \geq s,\, i \le n\}$ in the fashion of \cite{Koul1981}, following the lines of \cite{vanderLaan1998}. Specifically,
\begin{align} \label{eq:est_E}
  &\hat{E}_n(a,u,s,k) = \mathbb{P}_{n}[Y1(X \geq s, A=a)]\\
  & \quad + \frac{{\rm Cov}_{\mathbb{P}_{n}}(B_k1(X \geq s, A=a), Y1(X \geq s,A=a))}{{\rm Var}_{\mathbb{P}_{n}}(B_k1(X \geq s,A=a))}\big(u-\mathbb{P}_{n}[B_k1(X \geq s,A=a)]\big) \nonumber.
\end{align}
Note that for the given $k$ and each $(a,u)$, the process $s \mapsto \hat{E}_j(a,u,s,k)$ is unpredictable with respect to the filtration defined in \eqref{eq:filtration}. This is the unpredictability issue referred to in \cite{Huang2023}.
In the sequel we suppress the argument $s$ if $s=-\infty$ and use 
\begin{align} \label{eq:est_E.variant}
  \hat{E}_n(a,u,k) = \mathbb{P}_{n}[Y1(A=a)]
  + \frac{{\rm Cov}_{\mathbb{P}_{n}}(B_k1(A=a), Y1(A=a))}{{\rm Var}_{\mathbb{P}_{n}}(B_k1(A=a))}\big(u-\mathbb{P}_{n}[B_k1(A=a)]\big)
\end{align}
to denote the estimator of $E_{P}[\tilde{Y}\,|\,A=a,B_k=u] = E_P[T\,|\,A=a,B_k=u]$.

In terms of $f^*_k$ in \eqref{eq:if_star} with $P$ replaced by $\hat{P}_n$,
the proposed one-step estimator is now given by
\begin{align}\label{eq:if_onestep}
  S_k(\mathbb{P}_n, \hat{P}_n) = \Psi_k(\hat{P}_n) + \mathbb{P}_n f^*_k(O\,|\,\hat{P}_n),
\end{align}
where $\Psi_k(\hat{P}_n) = \hat{\beta}_{nk}\hat{\zeta}_{nk}$ was defined at the end of Section \ref{sec:NIE} and we have used plug-in of the various features {\bf \ref{item:estimation_E} {\rm --} \ref{item:estimation_reciprocal.odds}} to estimate $f^*_k$.

\section{Stabilized one-step estimator}
\label{sec:stabilized.one-step}
For estimating the target parameter $\Psi(P)$ defined in \eqref{eq:def_parameter}, we need to incorporate the selection of the most informative mediator into the one-step estimator \eqref{eq:if_onestep}.
Following the stabilization approach of \cite{Huang2023}, the idea is first to randomly order the data, and then consider subsamples consisting of the first $j$ observations for $j=q_n,\ldots,n-1$, where $\{q_n\}$ is some positive integer sequence such that both $q_n$ and $n-q_n$ tend to infinity. 
Based on the subsample of size $j$, the label of the most informative mediator is estimated by
\begin{align} \label{eq:predictor_selection}
  k_j = \arg \max_{k = 1,\ldots,p} \big|\Psi_{k}(\hat{P}_{nj})\big|.
\end{align}
The stabilized one-step estimator of $\Psi(P)$ is then given by
\begin{equation} \label{eq:Sn_star}
  S^*_{n} = \frac{1}{n-q_n}\sum_{j=q_n}^{n-1}w_{nj}m_jS_{k_j}(\delta_{O_{j+1}}, \hat{P}_{nj})\, ,             
\end{equation}
where $m_j \in \{-1,1\}$ is the sign of $\Psi_{k_j}(\hat{P}_{nj})$, $S_{k_j}$ refers to \eqref{eq:if_onestep} with the mediator $B_k$ now being $B_{k_j}$, and $\hat{P}_{nj}$ refers to $\hat{P}_n$ based on only the first $j$ observations to estimate part of the parameters of $P$.
Here $\delta_{O_{j+1}}$ is the Dirac measure putting unit mass at $O_{j+1}$, $w_{nj} := \bar{\sigma}_n/\hat{\sigma}_{nj}$ with $\bar{\sigma}_n=\{(n-q_n)^{-1}\sum_{j=q_n}^{n-1}(1/\hat{\sigma}_{nj})\}^{-1}$,  
\begin{align*}
  \hat{\sigma}^2_{nj} = \frac{1}{j}\sum_{i=1}^j\bigg\{f^*_{k_j}(O_i\,|\,\hat{P}_{nj})
  -\frac{1}{j}\sum_{i=1}^jf^*_{k_j}(O_i\,|\,\hat{P}_{nj})
  \bigg\}^2,                                
\end{align*}
and $f^*_{k_j}$ is $f^*_k$ with the mediator taken as $B_{k_j}$.

Under the following conditions, we establish the asymptotic normality of $S_n^*$.
\begin{assumpenum}[resume=assumptions]
  \item \label{assump:Mediator} Each mediator $B_k$ has bounded support.
  
  \item \label{assump:Survival function} The survival function of the censoring, $G$, is continuous and $G(\tau) > 0$.
  
  \item \label{assump:At-risk prob} There is a positive probability of a subject still being at risk at the end of follow-up: ${\rm P}(X \geq \tau) > 0$.
  
  \item \label{assump:Variances} ${\rm Var}(B_k)$, ${\rm Var}(B_k1(X \geq s, A=a))$ and ${\rm Var}(f^*_k(O\,|\,P))$ are bounded away from zero and infinity.
\end{assumpenum}

\begin{theorem}\label{Thm:stab_one_step}
Suppose the number of predictors $p=p_n$ satisfies
$\log(p_n)/n^{1/4} \to 0$, and the subsample size $q_n$ used for stabilization satisfies $n-q_n\to \infty$, $n\big/q_n = O(1)$, and $q_n^{1/4}\big/\log( n \lor p_n) \to \infty$.
Assume \ref{assump:Consistency}--\ref{assump:Variances}, the asymptotic stability conditions \ref{assump:Conditional_mean_E0}--\ref{assump:Signal strength} that are stated just before the proof of lemmas in the Appendix.
Then $S_n^*$ is an asymptotically normal estimator of $\Psi(P) \colon$
\begin{align*}
\sqrt{n-q_n}\bar{\sigma}_n^{-1}[S^*_n-\Psi(P)] \lcrarrow{d} \mathcal{N}(0,1).
\end{align*}
\end{theorem}

The following $100(1-\alpha)\%$ confidence interval for $\Psi(P)$ is justified by the above asymptotic normality:
\begin{align*} 
  [{\rm LB}_n, {\rm UB}_n]=\left[S^*_n \pm z_{\alpha/2}\,\frac{\bar{\sigma}_n}{\sqrt{n-q_n}}\right],
\end{align*} 
and the corresponding two-sided p-value   for testing the null hypothesis that $\Psi(P)=0$ is
$$2\big(1-\Phi\big(\,\big|\sqrt{n-q_n}S_n^*\big/\bar{\sigma}_n\big|\,\big)\big),
$$ where $\Phi$ is the cumulative distribution function of $\mathcal{N}(0,1)$, and $z_{\alpha/2}$ is the upper $\alpha/2$ quantile of $\Phi$.

It should be noted that the application of the stabilized one-step estimator needs to first randomize the order of the data, and to mitigate the effects of a single random ordering it is advisable to combine the results from say 100 random orderings (as we do in the sequel). The value of $q_n$ plays the role of a tuning parameter. Taking a smaller $q_n$ (holding $n$ and $p_n$ fixed) is expected to reduce variability in the performance of $S_n^*$, but taking a too-small value of $q_n$ leads to overfitting.
In practice, we recommend setting $q_n = 0.8n$ (which satisfies the conditions in Theorem \ref{Thm:stab_one_step}) as a reasonable trade-off, although in practice it is advisable to run the analysis for a few different values of $q_n$ and compare the results.

\section{Competing methods}
\label{sec:competing_methods}

\noindent {\bf Bonferroni-corrected one-step estimator:}
With the label $k$ of the strongest mediator $B_k$ estimated by $\hat{k}_n := \arg \max_{k = 1,\ldots,p}|\Psi_{k}(\hat{P}_n)|$ 
based on
the full sample, the test statistic is the standardized $S_{\hat{k}_n}(\mathbb{P}_n,\hat{P}_n)$. When disregarding the selection (i.e., viewing $\hat{k}_n$ as fixed), the null distribution of this statistic is asymptotically standard normal.
We apply the Bonferroni correction to the resulting p-value (or the $\alpha$-level of the confidence interval) and expect conservative behaviors for large $p$. Without the Bonferroni correction, this approach is anti-conservative and showcases the consequence of ignoring the selection bias that results from a naive use of $\hat{k}_n$.

\noindent {\bf Oracle one-step estimator:}
In this case, the label $k$ of the most contributing mediator $B_k$ is given, and the test statistic is simply the standardized $S_{k}(\mathbb{P}_n,\hat{P}_n)$, which has an asymptotically standard normal null distribution. Assuming knowledge of $k$ is of course unrealistic and thus this estimator cannot be used in practice, but this estimator serves as a benchmark in simulation studies for comparison purposes.

\section{Simulation studies}
\label{sec:simulation}

In this section, we report the results of simulation studies evaluating the performance of the stabilized one-step estimator (with $q_n=0.8n$) in comparison with the competing methods in Section \ref{sec:competing_methods}.
The treatment or exposure variable $A$ follows $Bernoulli(0.5)$.
The noise $\varepsilon$ for the outcome model is distributed as $\mathcal{N}(0,1)$ (independently of $\bs{B}$).
The log-transformed survival times are generated under one of the following AFT scenarios:
\begin{itemize}
\itemsep=-1pt
\item[] {\bf Model 0:} $\; T = 0.2\, A + \varepsilon$, and the mediators $\bs{B} = \bs{E}$, where $\bs{E}$ has a p-dimensional  normal distribution with unit variances and an exchangeable correlation structure with each pairwise correlation  $0.5$; 
\item[] {\bf Model 1:} $\; T = 0.4\,A + 0.2\,B_1 + \varepsilon$, and the mediators $\bs{B}$ with $B_1 = A + E_1$, $B_k = 0.6\,A + E_k$ for $2 \le k \le 5$, $B_k = 0.3\,A + E_k$ for $6 \le k \le 10$ and $B_k = E_k$ for $k \ge 11$, where each component of $(E_{1},\ldots,E_{10})$ follows $\mathcal{N}(0,1)$;
$(E_{11},\ldots, E_p)$ have a $(p-10)$-dimensional  normal distribution with unit variances and an exchangeable correlation structure with each pairwise correlation 0.1;
\item[] {\bf Model 2:} $\; T = 0.4\,A +\sum_{k=1}^p\beta_kB_k+\varepsilon$ with $\beta_1=\ldots=\beta_5=0.2$, $\beta_6=\ldots=\beta_{10} = -\,0.1$, $\beta_k=0$ for $k \geq 11$, and the same structure of mediators $\bs{B}$ as specified in {\bf Model 1}. 
\end{itemize}
In Model 0, neither natural direct nor indirect effects are present. In Model 1, the exposure is only mediated through a single active mediator ($k=1$), and the maximal NIE is 0.2. In Model 2, there are ten active mediators, the first being the most influential with maximal NIE again taking the value 0.2. 
The censoring time $C$ is taken to be the logarithm of an exponential random variable with a rate parameter that gives moderate censoring ($20\%$). For each data-generating scenario, we fix the sample size at $n=800$ and consider 5 values of $p$ of the form $10^a$ (for $a=2,3, \ldots, 6$). 
The Kaplan--Meier estimator $\hat{G}_n$ is used in $S_n^*$, as justified by the independent censoring assumption; although a more sophisticated conditional Kaplan--Meier estimator could be used instead, doing so would involve an additional computational cost.

Based on the proposed confidence interval, empirical coverage probabilities of $\Psi(P)$  in  Models 0--2, using $1000$ Monte Carlo replications in each case, are displayed in Figure \ref{fig:empirical_coverage_censoring20_wholesample}. We use the full-sample-based $\hat{P}_n$ to estimate the features of $P$ used in the stabilized one-step estimator. Figure \ref{fig:CIwidth_censoring20_wholesample} presents the resulting average confidence interval widths.
The panels for Models 0--2 show coverage probabilities in the corresponding models, respectively, with the nominal level of $90 \%$ shown by the horizontal black dashed line. The results for the Bonferroni-corrected one-step estimator are highly conservative, as expected, along with the increasing confidence interval widths as $p$ grows.
Throughout with the shortest confidence intervals among all the methods,
the Oracle one-step estimator provides a fair benchmark in Models 1--2 but has over-coverage in Model 0 where no active mediators are present.
The stabilized one-step estimator provides the closest-to-nominal coverage and stability in confidence interval width throughout, apart from the Oracle one-step estimator in Models 1--2.
Figure \ref{fig:QQ_plot} shows good agreement with the asymptotic normality of the stabilized one-step estimator for $p = 10^5$ and $10^6$.

\begin{figure}[htbp]
\begin{center}
 \includegraphics[width=\textwidth]{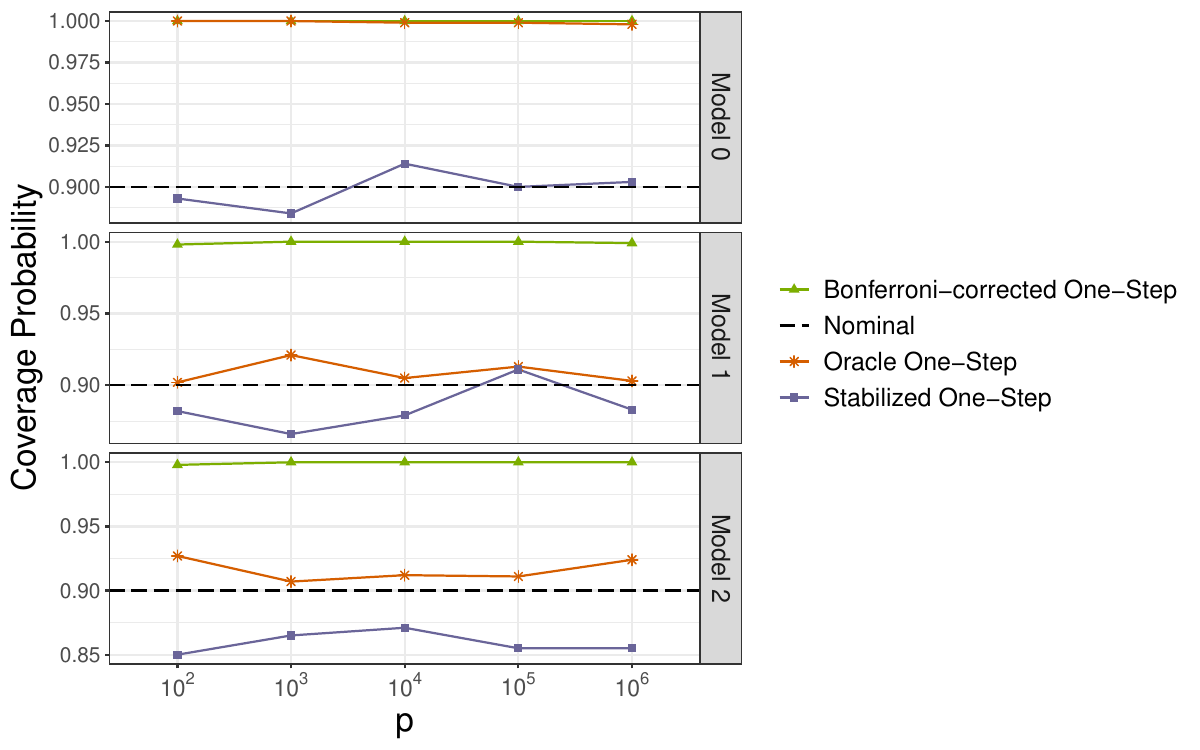}
 \caption{Empirical coverage probabilities (at nominal level 90\%) based on $1000$ samples ($n=800$) generated from Models 0--2 under moderate censoring ($20 \%$), for $p$ in the range $10^2$--$10^6$, using the full sample to obtain $\hat{P}_n$.}

\label{fig:empirical_coverage_censoring20_wholesample}
 \end{center}
\end{figure}

\begin{figure}[htbp]
\begin{center}
 \includegraphics[width=\textwidth]{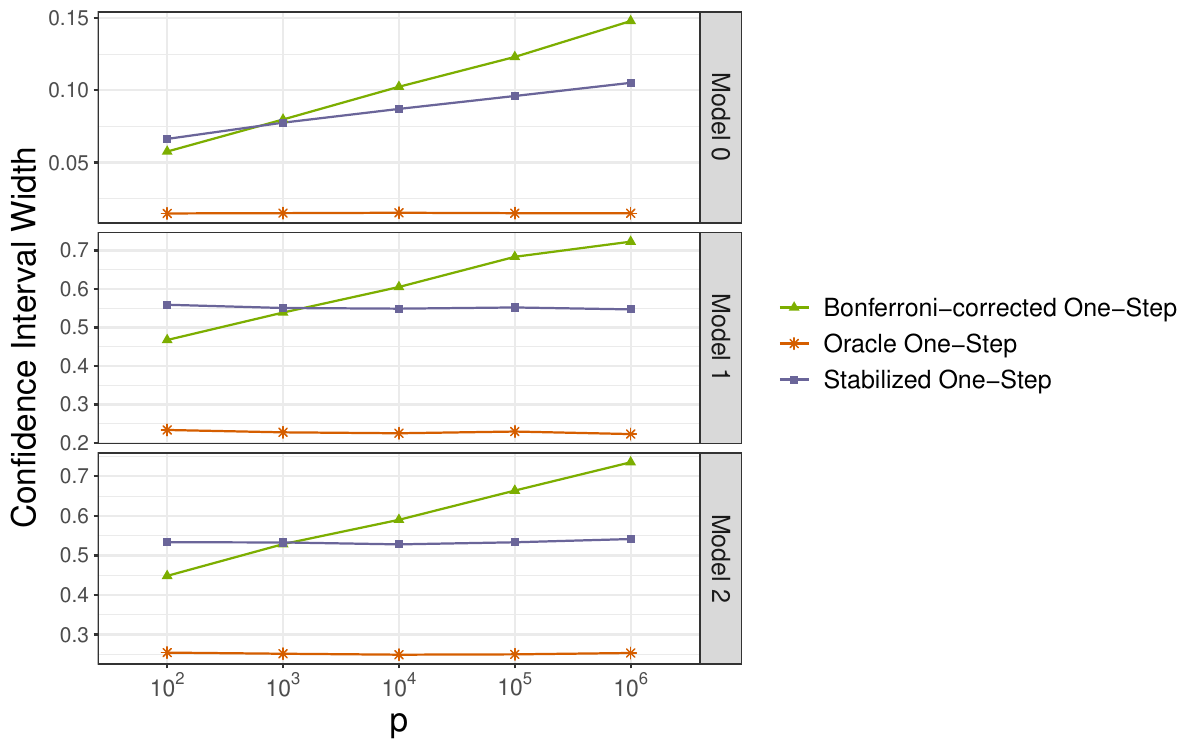}
 \caption{Average 90\% confidence interval width based on $1000$ samples ($n=800$) generated from Models 0--2 under moderate censoring ($20 \%$), for $p$ in the range $10^2$--$10^6$.}
\label{fig:CIwidth_censoring20_wholesample}
 \end{center}
\end{figure}

\begin{figure}[htbp]
\begin{center}
 \includegraphics[width=\textwidth]{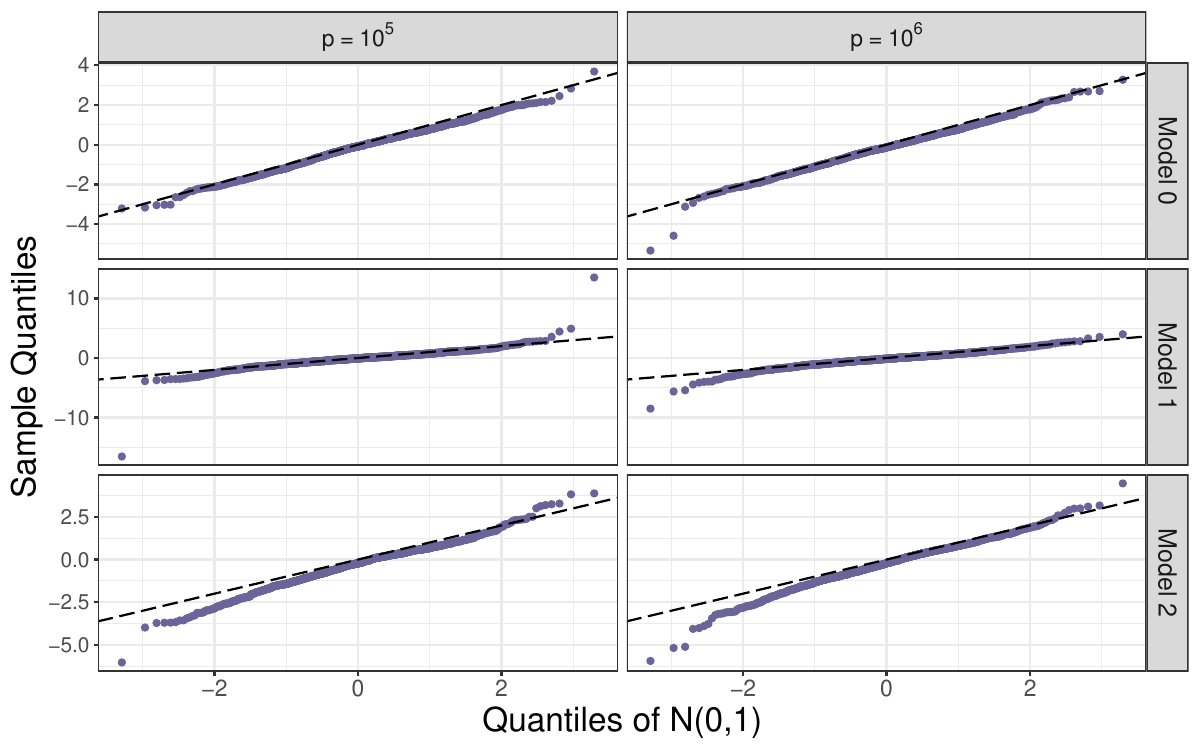}
 \caption{{\normalsize Empirical standardized test statistics based on $1000$ samples ($n=800$, $p = 10^5$ and $10^6$) plotted against standard normal quantiles.}}
\label{fig:QQ_plot}
 \end{center}
\end{figure}

\section{Confounding}
\label{sec:confounding}

In many biomedical studies with survival endpoints, various covariates can confound the relationship between mediators and the survival outcome. To examine the sensitivity of our approach to such confounding, suppose that we are given a (fully-observed) low-dimensional 
confounder $Z$. 
We should then estimate $\beta_k$ in \eqref{eq:def_parameter_k} by the least-squares estimator of the effect of the mediator $B_k$ on $T$ adjusted for $Z$, as in \cite{Koul1981}, which accordingly gives an {\it extended} version of the stabilized one-step estimator. Similarly, we can extend the Bonferroni-corrected and Oracle one-step estimators. The estimator of $\zeta_k$ could also readily be adjusted for $Z$.

In this framework, we report the results of simulation studies comparing the performance of the stabilized one-step estimator and its extended version (with $q_n=0.8n$) to the competing methods in Section \ref{sec:competing_methods} and their extended versions.
The simulation models are constructed using the parts of Models 0--2 for $T$ (as defined previously) with the inclusion of an independent $Z \sim Bernoulli(0.4)$, where $Z$ is independent of $(A, {\bs B}, \varepsilon)$ as follows: 
\begin{itemize}
\itemsep=-1pt
\item[] {\bf Model 0$'$:} $\; T = 0.2\,A - 0.1\,Z + \varepsilon$; 
\item[] {\bf Model 1$'$:} $\; T = 0.4\,A - 0.1\,Z + 0.2\,B_1 + \varepsilon$;
\item[] {\bf Model 2$'$:} $\; T = 0.4\,A - 0.1\,Z  +\sum_{k=1}^p\beta_kB_k+\varepsilon$ with $\beta_1=\ldots=\beta_5=0.2,$ $\beta_6=\ldots=\beta_{10} = -\,0.1,$ $\beta_k=0$ for $k \geq 11$. 
\end{itemize}

Compared to the results without the confounder involved, as in Figure \ref{fig:empirical_coverage_censoring20_wholesample}, empirical coverage probabilities based on $1000$ Monte Carlo replications remain close to the nominal level under the three extended models, either extended with the confounder adjusted in the estimation of natural indirect effects or not (Figure \ref{fig:empirical_coverage_censoring20_wholesample_ext}). In addition, we present the resulting average confidence interval widths in Figure \ref{fig:CIwidth_censoring20_wholesample_ext}, which shows including confounders does not increase the confidence interval width in comparison with the results in Figure \ref{fig:CIwidth_censoring20_wholesample}.
Figure \ref{fig:QQ_plot_ext} indicates that with the presence of confounders, the asymptotic normality of the stabilized one-step estimator is fairly well maintained. 

On the other hand, if the effect of the confounding is strong, then the results may well be very different between the unadjusted and extended approaches, as we see in the next section, in which case it would be advisable to rely on the extended approach. 
Another important sensitivity issue would be possible violation of the sequential ignorability assumption, as discussed in \cite{Tchetgen2012}, although we do not pursue that issue here.

\begin{figure}[htbp]
\begin{center}
 \includegraphics[width=\textwidth]{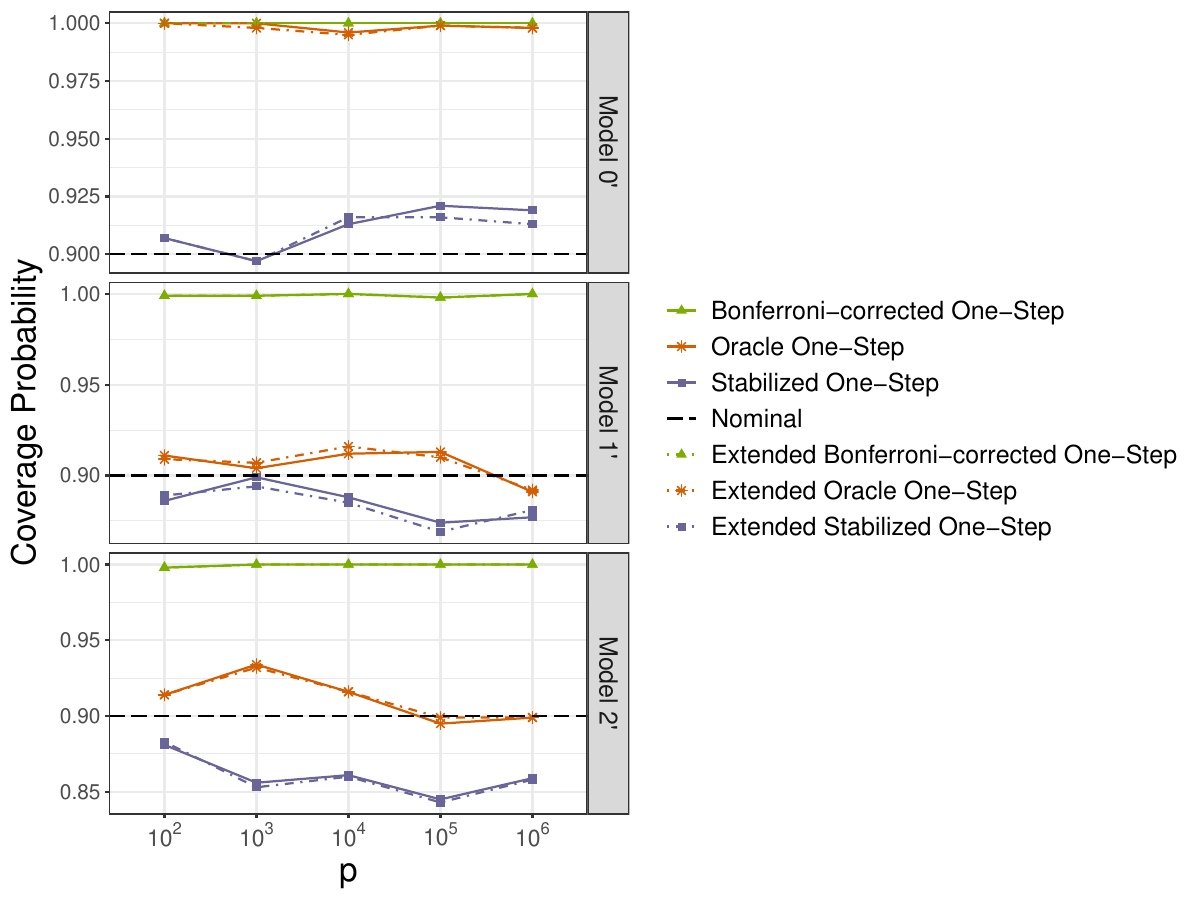}
 \caption{{\normalsize As in Figure \ref{fig:empirical_coverage_censoring20_wholesample}, except for empirical coverage probabilities based on survival times  generated from Models $0'$--$2'$.}}
\label{fig:empirical_coverage_censoring20_wholesample_ext}
 \end{center}
\end{figure}

\begin{figure}[htbp]
\begin{center}
 \includegraphics[width=\textwidth]{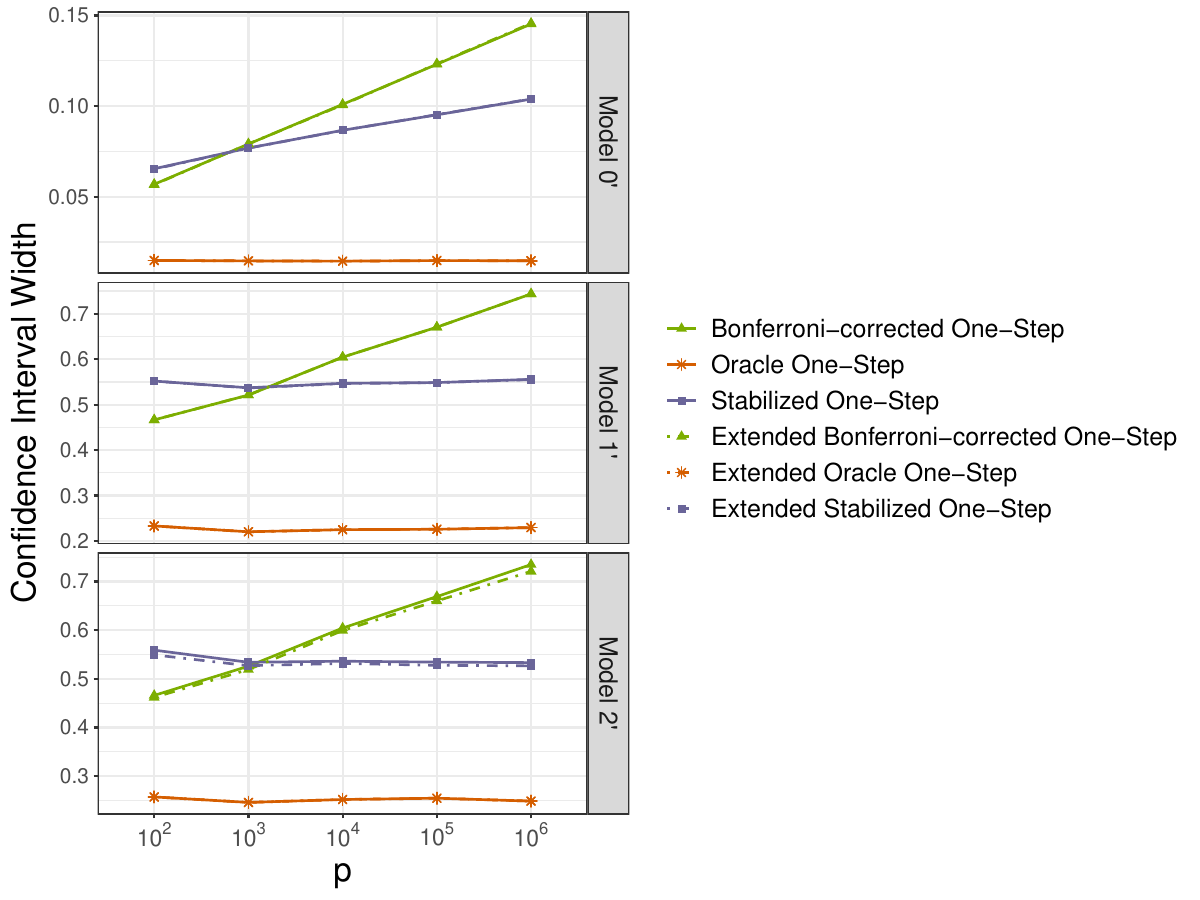}
 \caption{{\normalsize As in Figure \ref{fig:CIwidth_censoring20_wholesample}, except for average confidence interval width based on data generated from Models $0'$--$2'$.}}
\label{fig:CIwidth_censoring20_wholesample_ext}
 \end{center}
\end{figure}

\begin{figure}[htbp]
\begin{center}
 \includegraphics[width=\textwidth]{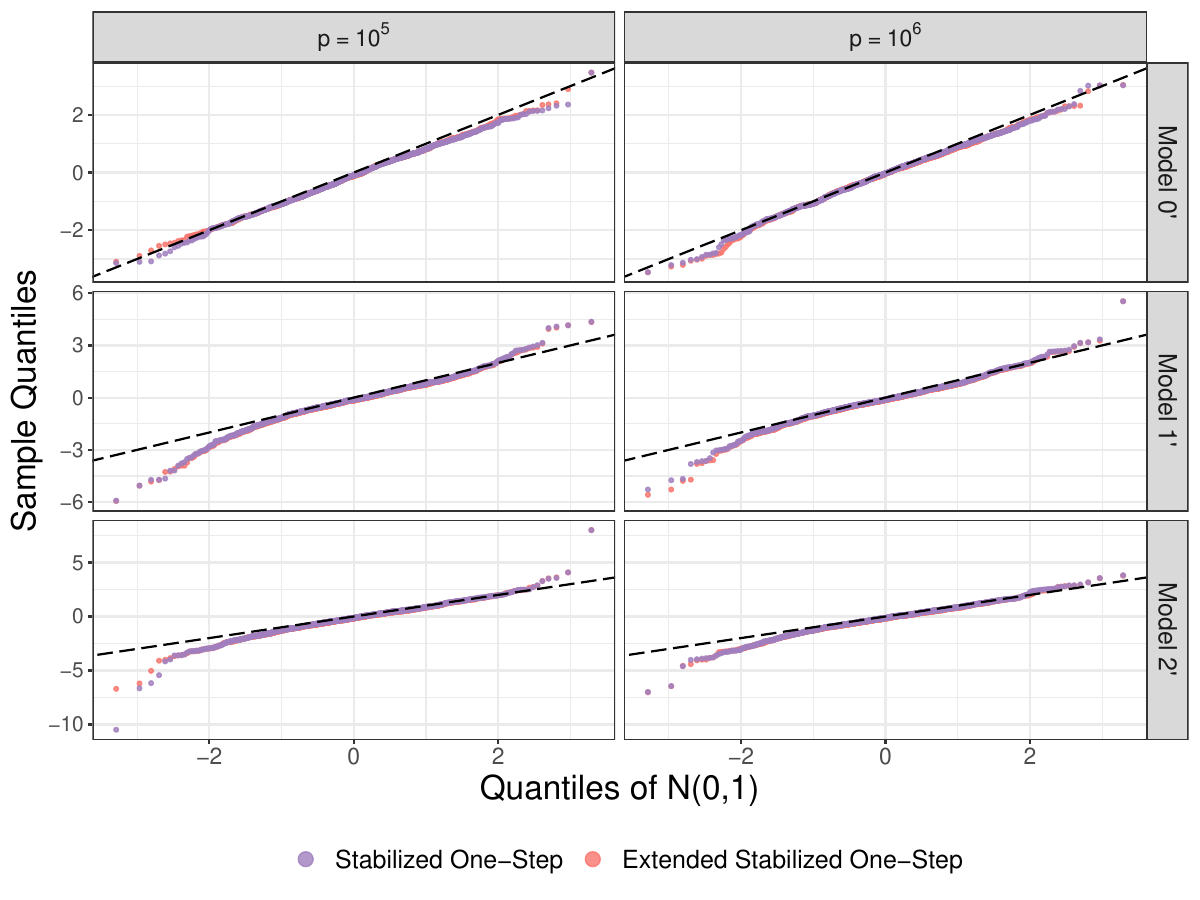}
 \caption{{\normalsize As in Figure \ref{fig:QQ_plot}, except for empirical standardized test statistics based on data generated from the extended models.}}
\label{fig:QQ_plot_ext}
 \end{center}
\end{figure}

\section{Application to lung cancer data}
\label{sec:lungcancer}
Lung cancer is one of the most prevalent types of cancer and is the leading cause of mortality worldwide \cite{sung2021global}. As reported by \cite{CAO20181483}, approximately 85\% of lung cancer cases are classified as non-small cell lung cancer (NSCLC), while the remaining 15\% are categorized as small-cell lung cancer.  \cite{breitling2011tobacco} found that tobacco smoking, an important risk factor for lung cancer, has been associated with changes in DNA methylation. 

As noted in the Introduction, DNA methylation is a reversible process \cite{wu2014reversing}, bringing considerable scientific interest to identify potential DNA methylation CpG sites that mediate the effect of smoking on the survival of lung cancer patients. We apply our proposed methods (the stabilized one-step estimator and its extended version) to analyze a cohort from the Cancer Genome Atlas (TCGA) project dataset (\url{https://xenabrowser.net/datapages/}), consisting of $n=754$ lung cancer patients aged between 33 and 90 years, and involving $p=$ 365,306 DNA methylation CpG sites. The DNA methylation profiles were measured using the Illumina Infinium HumanMethylation 450 platform and recorded by BeadStudio software. The exposure variable $A$ is smoking status (current/ever smoker versus non-smoker). 
The survival of lung cancer patients is encoded as the number of days from the initial diagnosis to the death or the censoring time. The median survival time is 1,632 days; 305 deaths were observed over the follow-up period with a censoring rate of 60\%.  Following previous papers that analyzed this dataset \cite{luo2020high,zhang2021mediation,tian2022coxmkf}, we adopt the independent censoring assumption.

Relevant confounders are age (median 68 years), gender (58\% male), pathologic stage (taking ten values: I, IA, IB, II, IIA, IIB, III, IIIA, IIIB, IV), and whether receiving radiation therapy (12\%). 
We simplify the pathologic stage variable to take the value 0 for a mild stage, 1 for moderate and 2 for severe, accounting for 53\%, 29\% and 18\% of the patients, respectively.  
As noted earlier, in practice the stabilized one-step estimator (and its extended version that adjusts for confounding) should be applied to multiple random orderings of the data, to mitigate possible selection bias from a single random ordering. We combine the results from say 100 random orderings simply using a Bonferroni correction.

In Table \ref{tab:estNIE_qn_SOSE.std}, we report
the point estimate and corresponding standard error obtained from the particular random ordering that yields the minimal p-value 
among the 100 random orderings of the data, using either the stabilized one-step estimator or its extended version, for various choices of $q_n$. The time needed to handle a single random ordering at a given value of $q_n$ is about ten minutes on a powerful desktop computer.
The Bonferroni-corrected confidence intervals (with a nominal level of 0.001), and the Bonferroni-corrected p-values, are also displayed.
We find that when adjusting for the confounders, there exists at least one significant mediator (DNA methylation at some CpG sites) for the effect of smoking on survival from lung cancer, when we choose $q_n =$ 302, 377, 419, and 503. For each value of $q_n$, the selected CpG sites across five sample splits in the most significant random ordering are listed in Figure \ref{fig:rdo100selects_plt.std},
together with their frequency of being selected. The full results are shown in Figure \ref{fig:rdo100CIs_plt.std}, where the individual point estimates and confidence intervals obtained from all 100 random orderings are displayed.  
For comparison, the testing procedures involving the Bonferroni-corrected one-step estimator and its extended version are also utilized, but no significant results can be found using those approaches (results not shown). The CpG site cg25644150 is located in the CpG island in gene SAR1B which has been found related to the development of lung cancer \cite{chen2021sar1b}. The CpG site cg04889061 is located in gene JPH3 which is a lung cancer-associated gene and is also associated with smoking \cite{bruse2014increased}. Those discovered CpG sites warrant more future follow-up studies to better understand their biological functions.

\begin{table}[ht]
\centering
\caption{For various choices of $q_n$ (with the percentage of the original data in  parentheses), the table  gives the estimates and standard errors of the NIE corresponding to the minimal p-values among 100 random orderings of the data, 
along with the Bonferroni-corrected confidence intervals (C.I.) and p-values that are adjusted for the multiple random orderings.}
\resizebox{0.85\textwidth}{!}{%
\begin{tabular}[t]{p{1cm}cccccccccc} %
\toprule
\multirow{5}{*}{\parbox{1cm}{}} & &
\multicolumn{4}{c}{Stabilized One-Step} &  &
\multicolumn{4}{c}{Extended Stabilized One-Step} \\ 
\cmidrule{3-6} \cmidrule{8-11}
 $q_n\,(\%)$ & & {\centering Est.} & 
 {\centering S.E.} & {C.I.} & {P-Value}
 & & {\centering Est.} & {\centering S.E.} &
 {C.I.} & {P-Value} \\
\midrule
\parbox{1cm}{$302\,(40\%)$} &
 & $1.49$ & $0.62$ & $(-0.55, 3.53)$ & $1.000$ &   
 & $2.03$ & $0.53$ & $(0.30, 3.77)$ & $0.012$ \\
\parbox{1cm}{$377\,(50\%)$} &
 & $1.91$ & $0.69$ & $(-0.37, 4.19)$ & $0.573$ &   
 & $2.61$ & $0.69$ & $(0.36, 4.87)$ & $0.014$ \\
\parbox{1cm}{$419\,(55\%)$} &
 & $1.91$ & $0.73$ & $(-0.50, 4.33)$ & $0.922$ &   
 & $2.52$ & $0.66$ & $(0.36, 4.69)$ & $0.013$ \\
\parbox{1cm}{$503\,(66\%)$} &
 & $2.24$ & $0.82$ & $(-0.47, 4.95)$ & $0.664$ &   
 & $3.19$ & $0.83$ & $(0.47, 5.91)$ & $0.011$ \\ 
\parbox{1cm}{$629\,(83\%)$} &
 & $3.05$ & $1.15$ & $(-0.74, 6.84)$ & $0.802$ &   
 & $2.84$ & $1.06$ & $(-0.65, 6.34)$ & $0.749$ \\ 
\bottomrule
\end{tabular}
}
\label{tab:estNIE_qn_SOSE.std}
\end{table}

\begin{figure}[htbp]
\begin{center}
 \includegraphics[width=\textwidth]{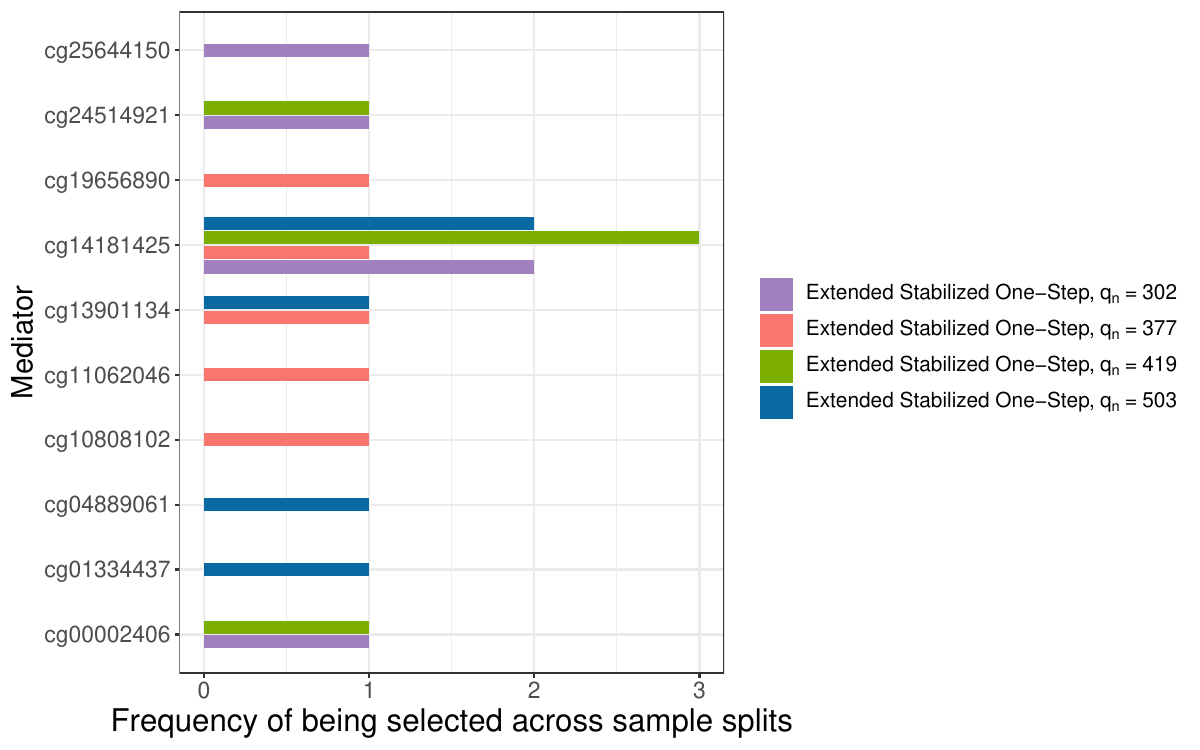}
 \caption{The selected mediators and their frequency of being selected across five sample splits in the most significant random ordering of data for each $q_n$, in which mediators are standardized.}
\label{fig:rdo100selects_plt.std}
 \end{center}
\end{figure}

\begin{figure}[htbp]
\begin{center}
 \includegraphics[width=\textwidth]{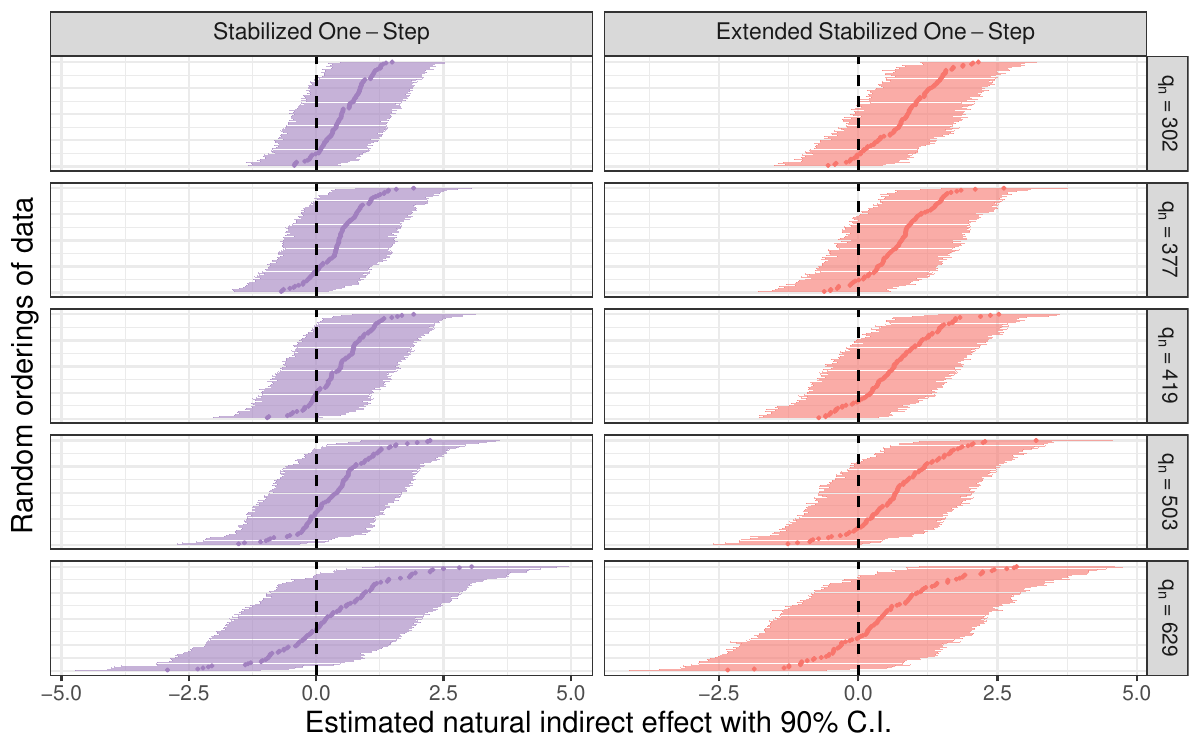}
 \caption{90\% confidence intervals (C.I.) from 100 random orderings of the lung cancer data for various values of $q_n$ along with the point estimates (dark dots), in which mediators are standardized.}
\label{fig:rdo100CIs_plt.std}
 \end{center}
\end{figure}

\section{Concluding remarks}

We have developed a post-selection inference method for the mediated effect of a binary exposure on a right-censored survival outcome and compared the numerical performance of our approach with competing methods of mediation analysis with high-dimensional mediators. This is done in terms of controlling the FWER and providing the asymptotic accuracy of a confidence interval for the natural indirect effect. In Chapter 4 of \cite{vanderweele2015explanation}, mediation analysis methods for survival outcome data have been developed, but mainly just for {\it identification}; beyond our proposed approach, little has been done in terms of formal statistical inference, especially in high-dimensional mediator settings. In the framework of
high-dimensional mediators and a survival outcome, there are only a few methods available, as reviewed in \cite{tian2022coxmkf}. Notably, \cite{tian2022coxmkf} use the knockoff method to control FDR for a survival outcome in finite-sample settings (without relying on large-sample theory) whereas, as mentioned in the Introduction, this method is not guaranteed to control FWER. Our proposed approach is the first to control FWER for inferring large-scale mediation effects with possibly right-censored survival outcomes. Given the increasing availability of genome-wide data in longitudinal follow-up studies, our method has great potential to help researchers better understand the causal pathways from exposure to survival outcomes.
One limitation of our methodology is that we assume independent censoring. Another limitation is that it is restricted to non-time-varying exposures, mediators, and confounders (because it is difficult to accommodate time-dependent covariates in the AFT framework). For future work, it would be of interest to extend our methods in such a direction.

\bibliographystyle{abbrvnat}
\bibliography{highdim_mediation}

\appendix
\renewcommand{\theequation}{S.\arabic{equation}}
\section*{Proof of Theorem 1}
We start by introducing a decomposition of the stabilized one-step estimator.
The distribution of $P$ is identified by $(Q, G, E_0)$, where $E_0(a,u,s,k) := E_{P}[\tilde Y\,|\,A=a, B_k=u, X \geq s]$. For $j \in \{q_n,\ldots,n\}$, define functions
$[\hat{E}_j - E_0] : (a,u,s,k) \mapsto \hat{E}_j(a,u,s,k) - E_0(a,u,s,k)$, 
$[\hat{Q}_{j}-Q] : (u,k) \mapsto \hat{Q}_{j}(u,k) -Q(u,k)$, and
$\hat{M}_{j}({\rm d}s) := 1(X_{j} \in {\rm d}s, \delta_{j}=0)-1(X_{j} \geq s)\, \hat{\Lambda}_n({\rm d}s)$.
Moreover, denote the probability $P(A=a)$ by $Q(a)$, for $a \in \mathcal{A}$.
Replacing in various ways each feature of $P$ by its estimator introduced gives $\hat{P}_{nj}=(\hat{E}_j, \hat{G}_n, \mathbb{Q}_{j}, \hat{q}_j, \hat{\beta}_{jk_j}, \hat{Q}_j)$;
$\hat{P}^{''}_{nj}=(E_0, \hat{G}_n, \mathbb{Q}_{j}, \hat{q}_j, \hat{\beta}_{jk_j}, \hat{Q}_j)$ and
$\hat{P}^{'''}_{nj}=(E_0, G, \mathbb{Q}_{j}, \hat{q}_j, \hat{\beta}_{jk_j}, \hat{Q}_j)$.
Therefore we are able to decompose the statistic of interest as
\begin{align} \label{eq:decomposition_rootn_Sn_star}
  &\sqrt{n-q_n}\bar{\sigma}_n^{-1}\big[S^*_n - \Psi(P)\big] \;=\; \frac{1}{\sqrt{n-q_n}}\sum_{j=q_n}^{n-1}\hat{\sigma}_{nj}^{-1}m_j\big[S_{k_j}(\delta_{O_{j+1}}, \hat{P}_{nj})-S_{k_j}(\delta_{O_{j+1}}, \hat{P}^{''}_{nj})\big] \\
  &\hspace{10pt} + \frac{1}{\sqrt{n-q_n}}\sum_{j=q_n}^{n-1}\hat{\sigma}_{nj}^{-1}m_j\big[S_{k_j}(\delta_{O_{j+1}}, \hat{P}^{''}_{nj})-S_{k_j}(\delta_{O_{j+1}}, \hat{P}^{'''}_{nj})\big] \nonumber\\
  &\hspace{10pt} + \frac{1}{\sqrt{n-q_n}}\sum_{j=q_n}^{n-1}\hat{\sigma}_{nj}^{-1}m_j\big[S_{k_j}(\delta_{O_{j+1}}, \hat{P}^{'''}_{nj})-S_{k_j}(\delta_{O_{j+1}}, P)\big] \nonumber\\
  &\hspace{10pt} + \frac{1}{\sqrt{n-q_n}}\sum_{j=q_n}^{n-1}\hat{\sigma}_{nj}^{-1}m_j\big[S_{k_j}(\delta_{O_{j+1}}, P)-\Psi_{k_j}(P)\big] \nonumber\\
  &\hspace{10pt} + \frac{1}{\sqrt{n-q_n}}\sum_{j=q_n}^{n-1}\hat{\sigma}_{nj}^{-1}\big[m_j\Psi_{k_j}(P)-\Psi(P)\big]
  := \mbox{(I) + (II) + (III) + (IV) + (V)}. \nonumber
\end{align}
Thus we can specifically have that
\begin{align} \label{eq:def_terms}
  \mbox{(I)} & =
  \frac{1}{\sqrt{n-q_n}}\sum_{j=q_n}^{n-1}\frac{m_j}{\hat{\sigma}_{nj}}\bigg\{
  - \frac{(1-A_{j+1})}{\mathbb{Q}_j(0)}\big[\hat{E}_j-E_0\big](1,B_{j+1,\,k_j},k_j) \nonumber\\
  & \hspace{2cm} + \frac{A_{j+1}}{\mathbb{Q}_j(0)}\hat{Q}_j(B_{j+1,\,k_j},k_j)\big[\hat{E}_j-E_0\big](1,B_{j+1,\,k_j},k_j) \nonumber\\
  &\hspace{2cm} + \frac{A_{j+1}}{\mathbb{Q}_j(1)}\bigg(1-\hat{Q}_j(B_{j+1,\,k_j},k_j)\frac{\mathbb{Q}_j(1)}{\mathbb{Q}_j(0)}\bigg)\int_{\mathcal{T}} \big[\hat{E}_j-E_0\big](A_{j+1},B_{j+1,\,k_j},s,k_j)\hat{M}_{j+1}({\rm d}s)\bigg\}, \nonumber \\
  \mbox{(II)} & = \frac{1}{\sqrt{n-q_n}}\sum_{j=q_n}^{n-1}\frac{m_j}{\hat{\sigma}_{nj}}\bigg\{\frac{A_{j+1}}{\mathbb{Q}_j(1)}\bigg(1 - 
  \hat{Q}_j(B_{j+1,\,k_j},k_j)\frac{\mathbb{Q}_j(1)}{\mathbb{Q}_j(0)}\bigg)\big(Y_{j+1}-\tilde{Y}_{j+1}\big) \nonumber \\
  & \hspace{1cm} - \frac{A_{j+1}}{\mathbb{Q}_j(1)}\bigg(1-\hat{Q}_j(B_{j+1,\,k_j},k_j)\frac{\mathbb{Q}_j(1)}{\mathbb{Q}_j(0)}\bigg) \int_{\mathcal{T}}E_0(A_{j+1},B_{j+1,\,k_j},s,k_j)1(X_{j+1} \ge s)[\hat{\Lambda}_n-\Lambda]({\rm d}s)\bigg\}, \nonumber \\
  \mbox{(III)} & = \frac{1}{\sqrt{n-q_n}}\sum_{j=q_n}^{n-1}\frac{m_j}{\hat{\sigma}_{nj}}\bigg\{
  - \bigg(\frac{1}{\mathbb{Q}_j(0)}-\frac{1}{Q(0)}\bigg)(1-A_{j+1})\Big[E_0(1,B_{j+1,\,k_j},k_j) - {\color {change}\hat{\beta}_{jk_j}\hat{q}_j(0,k_j)}\Big] \nonumber \\
  & \hspace{1.8cm} - \bigg(1 - \frac{1}{Q(0)}(1-A_{j+1})\bigg){\color{change}\Big[\hat{\beta}_{jk_j}\hat{q}_j(0,k_j)-\beta_{k_j}q(0,k_j)\Big]} \nonumber \\
  &\hspace{1.8cm} + \bigg(\frac{1}{\mathbb{Q}_j(1)} - \frac{1}{Q(1)}\bigg)A_{j+1}\Big[\tilde{Y}_{j+1} - {\color {change} \hat{\beta}_{jk_j}\hat{q}_j(1,k_j)} + \int_{\mathcal{T}} E_0(A_{j+1},B_{j+1,\,k_j},s,k_j)M_{j+1}({\rm d}s)\Big] \nonumber \\
  &\hspace{1.8cm} + {\color {change} \bigg(1 - \frac{1}{Q(1)}A_{j+1}\bigg)\Big[\hat{\beta}_{jk_j}\hat{q}_j(1,k_j)-\beta_{k_j}q(1,k_j)\Big]} \nonumber \\
  &\hspace{1.8cm} - \frac{1}{\mathbb{Q}_j(0)}A_{j+1}\big[\hat{Q}_j-Q\big](B_{j+1,\,k_j},k_j)\Big[\tilde{Y}_{j+1}-E_0(1,B_{j+1,\,k_j},k_j)\Big] \nonumber \\
  &\hspace{1.8cm} - \bigg(\frac{1}{\mathbb{Q}_j(0)}-\frac{1}{Q(0)}\bigg)A_{j+1}Q(B_{j+1,\,k_j},k_j)\Big[\tilde{Y}_{j+1}-E_0(1,B_{j+1,\,k_j},k_j)\Big] \nonumber \\
  & \hspace{1.8cm} - \bigg(\frac{1}{\mathbb{Q}_j(0)}-\frac{1}{Q(0)}\bigg)A_{j+1}\hat{Q}_{j}(B_{j+1,\,k_j},k_j)\int_{\mathcal{T}} E_0(A_{j+1},B_{j+1,\,k_j},s,k_j)M_{j+1}({\rm d}s) \nonumber \\
  & \hspace{1.8cm} -\frac{1}{Q(0)}A_{j+1}\big[\hat{Q}_{j}-Q\big](B_{j+1,\,k_j},k_j)\int_{\mathcal{T}} E_0(A_{j+1},B_{j+1,\,k_j},s,k_j)M_{j+1}({\rm d}s)\bigg\}.
\end{align}

Let $\tilde{\epsilon} > 0$; define $\mathcal{G}$ to be the collection of monotone nonincreasing c\`{a}dl\`{a}g functions $\tilde{G} \colon \mathcal{T} \rightarrow [0,1]$ such that $\tilde{G}(\tau) > \tilde{\epsilon}$.
Given constants $m_0, m_1 >0$, where the constants can be shown to exist following Lemma S6.2 of \cite{Huang2023}, let ${\cal BV}({\cal A} \times {\cal T})$ be the collection of functions $f \colon {\cal A} \times {\cal T} \rightarrow [-m_0,m_0]$ with total variation bounded by $m_1$.
Below we use the notation $\bs{b} =(b_1,\ldots,b_p)$. For $k\in\mathbb{N}$, $(q,v,w)\in\{0,1,2\}^3$, and $r\in\{0,1,2,3,4\}$, define the function classes
\begin{align}\label{eq:class_tilde_F}
  &\tilde{\mathcal{F}}_{1}(k,q,r) = \bigg\{(a, x, \delta, \bs{b})\: \mapsto\: b_k^{r}\Big(a\frac{\delta x}{\tilde{G}(x)}\Big)^q1(a=a', x \geq s):\: a' \in {\cal A},\: \tilde{G}\in\mathcal{G},\: s \in {\cal T} \bigg\}; \nonumber \\
  &\tilde{{\cal E}}(k,w) = \bigg\{(a',\bs{b},s)\: \mapsto\: 
  \Big[r_1(a',s) + r_2(a',s)b_k\Big]^w:\: r_1, r_2 \in {\cal BV}({\cal A} \times {\cal T}) \bigg\}; \nonumber \\
  &\tilde{\mathcal{F}}_{2}(k,v) = \bigg\{(a, x, \delta, \bs{b})\: \mapsto
  \Big[\int_{\mathcal{T}}\big[r_1(a',s) + r_2(a',s)b_k\big]1(x \geq s)\big\{1(x \in {\rm d}s, \delta=0) - \tilde{\Lambda}({\rm d}s)\big\}\Big]^v : \nonumber \\
  &\hspace{5cm} \tilde{\Lambda}(s) = -\log(\tilde{G}(s));\:
  \tilde{G}\in\mathcal{G};\: r_1,\, r_2 \in {\cal BV}({\cal A} \times {\cal T}) \bigg\}; \nonumber \\
  &\tilde{\mathcal{F}}(k,q,r,v,w) =
  \Big[\tilde{\mathcal{F}}_{1}(k,q,r)\tilde{{\cal E}}(k,w)\Big] \cup 
  \Big[\tilde{\mathcal{F}}_{1}(k,q,r)\tilde{\mathcal{F}}_{2}(k,v)\Big],
\end{align}
where for two function classes $\mathcal{H}_1$ and $\mathcal{H}_2$, we let $\mathcal{H}_1\mathcal{H}_2=\{h_1(\cdot)h_2(\cdot) : h_1\in\mathcal{H}_1,h_2\in\mathcal{H}_2\}$.
Also, for $k\in\mathbb{N}$, let $\tilde{\mathcal{F}}(k)=\cup_{q=0}^{2}\cup_{r=0}^{4} \cup_{v=0}^{2}\cup_{w=0}^{2}\tilde{\mathcal{F}}(k,q,r,v,w)$.
Following the arguments used for Lemmas S6.1-S6.5 of \cite{Huang2023}, we can show that
$\tilde{\mathcal{F}}(k)$ is a Vapnik-\v{C}ervonenkis (VC)-hull class for sets.

In the sequel, we use $\lesssim$ to denote ``bounded above up to a universal multiplicative constant that does not depend on $(j,n)$'' and $a \lor b:= \sup(a, b)$.
Suppose the number of predictors $p = p_n$ satisfies $\log(p_n)/n^{1/4} \to 0$, and the subsample size $q_n$ used for stabilization satisfies $n - q_n \to \infty$, $n\big/q_n = O(1)$, and $q_n^{1/4}\big/log(n \lor p_n) \to \infty$.
Let $\mathcal{K}_n = \{1, \ldots, p_n\}$. Henceforth we consider a large class 
$\tilde{\mathcal{F}}_n = \cup_{k\in\mathcal{K}_n}\tilde{\mathcal{F}}(k)$.
For $j =1,\ldots, n$, let $\mathbb{P}_j$ denote the empirical distribution of $O_1, \ldots, O_j$. With $K_{nj} := \{\log(n \lor p_n)\,/\,j\}^{1/2}$ for $j=q_n,\ldots,n$,
$\mathcal{I}_n := \{(a,k,s,u) : a \in \mathcal{A},\, k \in \mathcal{K}_n,\, s \in \mathcal{T},\, u \in \mathbb{R}\}$ and
$\mathcal{I}'_n := \{(a,j,k) : a \in \mathcal{A},\, j \in \{q_n,\ldots,n\},\, k \in {\cal K}_n\}$,
we define the following auxiliary events.
\begin{itemize}
  \item $\mathcal{A}_n := \cap_{j=1}^n\mathcal{A}_{nj}$, where $\mathcal{A}_{nj} := \Big\{\sup_{\tilde{f} \in \tilde{\mathcal{F}}_n} \big|(\mathbb{P}_j-P)\,\tilde{f}\big| \lesssim K_{nj}\Big\}$.
  
  \item Let $Y_n(\cdot) = \sum_{i=1}^n1(X_i \geq \cdot)$ and $\hat{\Lambda}_n(\cdot) = \int_{-\infty}^{\cdot}\big[1(Y_n(s)>0)/Y_n(s)\big]N_n({\rm d}s)$.
  \begin{align*}
    \mathcal{B}_n := \bigg\{&\sup_{s \in \mathcal{T}}\bigg|\frac{\hat{G}_n(s)}{G(s)}-1\bigg| \leq \sqrt{\frac{\log n}{n}},\:\: \sup_{s \in \mathcal{T}}\big|\hat{\Lambda}_n(s)-\Lambda(s)\big| \le \sqrt{\frac{\log n}{n}},\:\:\inf_{s \in \mathcal{T}}Y_n(s) \ge \sqrt{n},\\
    & \sup_{(a,k,s,u) \in {\cal I}_n}\big|\hat{E}_j(a,u,s,k)-E_0(a,u,s,k)\big| \lesssim K_{nj}\,\mbox{ for }\, j = q_n,\ldots,n \bigg\}.
  \end{align*}

  \item There exist $\varepsilon > 0$ and $\varepsilon' > 0$ such that
  \begin{align*}
    \mathcal{C}_n := \bigg\{
    & \sup_{(a,k,s,u) \in {\cal I}_n}\big|\hat{E}_j(a,u,s,k)\big| \lesssim K_{nj} + \sup_{(a,k,s,u) \in {\cal I}_n}\big|E_0(a,u,s,k)\big|,\; 
    \inf_{(a,j) \in \mathcal{A} \times \{q_n,\ldots,n\}}\mathbb{Q}_{j}(a) > \varepsilon,\\
    & \hspace{0.5cm} \sup_{a \in \mathcal{A}}\Big|1/\mathbb{Q}_{j}(a) - 1/Q(a)\Big| \lesssim 1/\sqrt{j},\;
    \sup_{(a,k) \in \mathcal{A} \times \mathcal{K}_n}\Big|\hat{q}_{j}(a,k) - q(a,k)\Big| \lesssim K_{nj} \\
    & \hspace{0.5cm} \sup_{k \in \mathcal{K}_n}\big|\hat{\beta}_{jk} - \beta_{k}\big| \lesssim K_{nj},\; \sup_{(k, u) \in \mathcal{K}_n \times \mathbb{R}}\big|\hat{Q}_{j}(u,k)-Q(u,k)\big| \lesssim K_{nj}\,\mbox{ for }\, j = q_n,\ldots,n \bigg\}.
  \end{align*}

  \item Provided $d_n(P_1,P_2) := P\big[\sup_{k \in \mathcal{K}_n}\big|\,f^*_{k}(O\,|\,P_1)-f^*_{k}(O\,|\,P_2)\big|1_{\mathcal{A}_n \cap \mathcal{B}_n \cap {\cal C}_n}\,\big]$ and \\
  $\hat{P}'_{nj} := (\hat{E}_j, G, \mathbb{Q}_{j}, \hat{q}_j, \hat{\beta}_{jkj}, \hat{Q}_j)$, define $\mathcal{D}_{n} = \cap_{j=q_n}^{n-1}\mathcal{D}_{nj}$, where
  \begin{align*}
    \mathcal{D}_{nj} := \Big\{d_n(\hat{P}_{nj}, \hat{P}_{nj}') \lor d_n(\hat{P}_{nj}',P) \lesssim K_{nj}\Big\}.
  \end{align*}  

  \item Let $\sigma^2_{nj} =\int f^*_{k_j}(o\,|\,P)^2 dP(o)$. There exists a constant $\varepsilon''$ such that
  \begin{align*}
    \mathcal{E}_n := \bigg\{&\min_{j \in \{q_n,\ldots,n\}} \hat{\sigma}_{nj}^2 > \varepsilon'',\; \Big|\frac{\sigma_{nj}}{\hat{\sigma}_{nj}} - 1\Big| \lesssim K_{nj}\,\mbox{ for }\, j = q_n,\ldots,n\bigg\}.
  \end{align*}

  \item Let $S(\tau) = {\rm P}(T \ge \tau)$.
  Define $\mathcal{H}_{n} = \cap_{j=q_n}^{n-1} \mathcal{H}_{nj}$, where
\begin{align*}
  \mathcal{H}_{nj} := \bigg\{\sup_{s \in \mathcal{T}}\bigg|\frac{\hat{G}_j(s)}{G(s)}-1\bigg| \le \frac{1}{S(\tau)}\sqrt{\frac{\log n}{j}}\: \bigg\}.
\end{align*}
 Note that $S(\tau)>0$, following that $G(\tau)>0$ by \ref{assump:Survival function}, that ${\rm P}(X \ge \tau) > 0$ by \ref{assump:At-risk prob}, and the independent censoring assumption that implies ${\rm P}(X \ge \tau)=S(\tau)G(\tau)$.
\end{itemize}
Assembling all the above auxiliary events gives ${\cal L}_n =
\mathcal{A}_n \cap \mathcal{B}_n \cap \mathcal{C}_{n} \cap \mathcal{D}_{n} \cap \mathcal{E}_{n} \cap \mathcal{H}_{n}$ and by Lemma \ref{lemma:Ln_prob_to_one} below, we can show that
\begin{align}
    {\rm P}({\cal L}_n)\to 1 \label{eq:Ln1}
\end{align}
when the conditions of Theorem \ref{Thm:stab_one_step} hold.
By \eqref{eq:Ln1}, it suffices to show the asymptotically negligibility of {\rm (I),\; (II),\; (III)} and {\rm (V)} after multiplication by $1_{{\cal L}_n}$, which are postponed to Lemmas \ref{lemma:asymp.I}--\ref{lemma:asymp.V}.

Replacing $(\mathbb{P}_n,\hat{P}_n)$ in the expression in \eqref{eq:if_onestep} by $(\delta_{O_{j+1}}, P)$ and with the predictor taken as $B_{k_j}$ gives that $S_{k_j}(\delta_{O_{j+1}}, P) = \Psi_{k_j}(P) + f^*_{k_j}(O_{j+1}\,|\,P)$. This further implies that
\begin{align*}
  \mbox{(IV)} = \frac{1}{\sqrt{n-q_n}}\sum_{j=q_n}^{n-1}\hat{\sigma}_{nj}^{-1}m_j\,f^*_{k_j}(O_{j+1}\,|\,P),    
\end{align*}
where $f_{k_j}^*(\cdot\,|\, P)$ is $f^*_k(\cdot\,|\, P)$ with the predictor $B_k$ taken as $B_{k_j}$, and the asymptotic normality of (IV) follows the martingale central limit theorem for triangular arrays (e.g., Theorem 2 in \cite{Gaenssler1978}) and the arguments in Appendix C of \cite{Huang2023}. 

\subsection*{Proofs of lemmas}
\label{app-sec:lemmas}
 Before proceeding to the required lemmas and their proofs we make the following asymptotic stability assumptions:
\begin{assumpenum}[resume=assumptions]
  \item \label{assump:Conditional_mean_E0}
  $\hat{E}_n$ defined in \eqref{eq:est_E} 
  for a given mediator $B_k$ consistently estimates (pointwise in its arguments) the true conditional mean residual life function $E_0(a,u,s,k):= E_P[\tilde{Y}\,|\, A=a, B_k=u, X \ge s]$, which is assumed to be uniformly-bounded and left-continuous in $s$, and with $k_j$ defined in \eqref{eq:predictor_selection},
  \begin{align*}
   E\bigg[\sup_{(j,s) \in \{q_n,\ldots,n\}
  \times {\cal T}}\big|E_0(A,B_{k_j},s,k_j)-E_0(A,B_{k_{j-1}},s,k_{j-1})\big|\bigg] = o(n^{-1/2}).
  \end{align*}
  
  \item \label{assump:Signal strength} If the parameter $\Psi_k(P)\neq 0$ for some $k$, there exists a sufficiently large $c > 0$ and a sequence of non-empty subsets $\mathcal{K}_n^* \subseteq \mathcal{K}_n=\{1,\ldots,p_n\}$ such that
  \begin{align*}
    \inf_{k \in \mathcal{K}_n^*}\big|\Psi_k(P)\big| - \sup_{l \in \mathcal{K}_n \setminus \mathcal{K}_n^*}\big|\Psi_l(P)\big| \ge c\sqrt{\frac{\log(n \lor p_n)}{q_n}},
  \end{align*}
  where the supremum over $l \in \mathcal{K}_n \setminus \mathcal{K}_n^*$ is defined to be 0 if $\mathcal{K}_n^*=\mathcal{K}_n$, and
  \[
  {\rm Diam}(\mathcal{K}^*_n) := \sup_{k, l \in \mathcal{K}^*_n}\big|\,\big|\Psi_k(P)\big|-\big|\Psi_l(P)\big|\,\big| = o(n^{-1/2}).
  \]
\end{assumpenum}

\begin{lemma} \label{lemma:An_Bn_prob_to_one}
For any sample size $n$, the event $\mathcal{A}_n$ occurs with probability at least $1-1/n$. Under the conditions of Theorem \ref{Thm:stab_one_step}, 
${\rm P}(\mathcal{A}_n\,\cap\, \mathcal{B}_n) \to 1$.
\end{lemma}
\begin{proof}
The proof can be completed, following Lemmas S6.6-S6.7 of \cite{Huang2023}.
\end{proof}

\begin{lemma} \label{lemma:An_Bn_Cn_prob_to_one}
Under the conditions of Theorem \ref{Thm:stab_one_step}, the event $\mathcal{C}_n$ is the intersection of
\begin{lemmaitemenum}
  \item \label{eq:event_Cn_1}
  $\sup_{(a,k,s,u) \in {\cal I}_n}\big|\hat{E}_j(a,u,s,k)\big| \lesssim K_{nj} + \sup_{(a,k,s,u) \in {\cal I}_n}\big|E_0(a,u,s,k)\big|$ for $j = q_n,\ldots,n$,
  \item \label{eq:event_Cn_2} $\inf_{(a,j) \in \mathcal{A} \times \{q_n,\ldots,n\}}\mathbb{Q}_j(a)$ is bounded away from zero,
  \item \label{eq:event_Cn_3} $\sup_{a \in \mathcal{A}}\big|1/\mathbb{Q}_j(a) - 1/Q(a)\big| \lesssim 1/\sqrt{j}$ for $j = q_n,\ldots,n$,
  \item \label{eq:event_Cn_4} $\sup_{(a,k) \in \mathcal{A} \times \mathcal{K}_n}\big|\hat{q}_{j}(a,k) - q(a,k)\big| \lesssim K_{nj}$ for $j = q_n,\ldots,n$,
  \item \label{eq:event_Cn_5} {\color {change} $\sup_{k \in \mathcal{K}_n}\big|\hat{\beta}_{jk} - \beta_{k}\big| \lesssim K_{nj}$ for $j = q_n,\ldots,n$,}
  \item \label{eq:event_Cn_6} {\color {change} $\sup_{(k, u) \in \mathcal{K}_n \times \mathbb{R}}\big|\hat{Q}_{j}(u,k)-Q(u,k)\big| \lesssim K_{nj}$ for $j = q_n,\ldots,n$.}
\end{lemmaitemenum}
where each of \ref{eq:event_Cn_1}, \ref{eq:event_Cn_3}, \ref{eq:event_Cn_4}, \ref{eq:event_Cn_5} and \ref{eq:event_Cn_6} relies on appropriately specified constants that do not depend on $(j,n)$. Then such constants exist such that
${\rm P}(\mathcal{A}_n \cap \mathcal{B}_n\cap {\cal C}_n) \to 1$.
\end{lemma}
\begin{proof}
According to \ref{assump:Conditional_mean_E0},
$\sup_{(a,k,s,u)}\big|E_0(a,u,s,k)\big|$ is uniformly bounded by some $(j,n)$-independent constant. Therefore when ${\cal B}_n$ occurs, we have $\sup_{(a,k,s,u)}\big|\hat{E}_j(a,u,s,k) - E_0(a,u,s,k)\big| \leq \tilde{K}K_{nj}$ for all $j$, implying that
\begin{align*}
  \sup_{(a,k,s,u)}\big|\hat{E}_j(a,u,s,k)\big| \leq
  \tilde{K}K_{nj} + \sup_{(a,k,s,u)}\big|E_0(a,u,s,k)\big| \lesssim K_{nj} + \sup_{(a,k,s,u)}\big|E_0(a,u,s,k)\big|.
\end{align*}
Thus, letting $\mathcal{C}_{n1}$ denote the event that
\begin{align*}
  \sup_{(a,k,s,u)}\big|\hat{E}_j(a,u,s,k)\big| \lesssim K_{nj} + \sup_{(a,k,s,u)}\big|E_0(a,u,s,k)\big| \mbox{ for } j=q_n,\ldots,n,
\end{align*}  
we have 
${\rm P}(\mathcal{A}_n \cap \mathcal{B}_n \cap {\cal C}_{n1}^c) \to 0$.

Following \ref{assump:Positivity}, there exists $\zeta>0$ such that $\min_{a \in \mathcal{A}}Q(a) \geq \zeta$. When ${\cal A}_{n}$ holds, then there exists $\eta>0$ such that for all $n \in \mathbb{N}$ and $a \in \mathcal{A}$, $\big|\mathbb{Q}_j(a)-Q(a) \big| \leq \eta\,j^{\,-1/2}$. This yields, for all $n$ and all $a \in \mathcal{A}$, 
\begin{align*}
  \mathbb{Q}_j(a) = Q(a)+\mathbb{Q}_j(a)-Q(a)
  \geq Q(a) - \big|\mathbb{Q}_j(a)-Q(a)\big|  \geq \zeta - \eta\,j^{\,-1/2}.
\end{align*}
Provided that $j^{\,-1/2} \to 0$ as $n$ grows sufficiently large, from the above display, we see that for all such $n$, $\mathbb{Q}_j(a) \geq \zeta$ when $\mathcal{A}_n \cap \mathcal{B}_n$ occurs.  Denote the event that $\inf_{(a,j)}\mathbb{Q}_j(a)\geq \zeta$ as $\mathcal{C}_{n2}$, and we show 
${\rm P}(\mathcal{A}_n \cap \mathcal{B}_n \cap {\cal C}_{n2}^c) \to 0$.

Let $\mathcal{C}_{n3}$ correspond to the event \ref{eq:event_Cn_3},
which is implied by $\mathcal{A}_{n} \cap \mathcal{B}_{n} \cap \mathcal{C}_{n2}$ and \ref{assump:Positivity}, thus giving
${\rm P}(\mathcal{A}_n \cap \mathcal{B}_n \cap \mathcal{C}_{n2} \cap \mathcal{C}_{n3}^c) \to 0.$
Similarly, let $\mathcal{C}_{n4}$ correspond to the event \ref{eq:event_Cn_4},
which is implied by $\mathcal{A}_{n} \cap \mathcal{B}_{n} \cap \mathcal{C}_{n2} \cap \mathcal{C}_{n3}$, \ref{assump:Positivity} and \ref{assump:Mediator}, thus giving ${\rm P}(\mathcal{A}_n \cap \mathcal{B}_n \cap \mathcal{C}_{n2} \cap \mathcal{C}_{n3} \cap \mathcal{C}_{n4}^c) \to 0.$

{\color {change} To show the event \ref{eq:event_Cn_5}, first note that
\begin{align}\label{eq:def.beta_k}
  \beta_{k} := \frac{{\rm Var}_{P}(A){\rm Cov}_{P}(B_k,\tilde{Y})-{\rm Cov}_{P}(A,B_k){\rm Cov}_{P}(A,\tilde{Y})}{{\rm Var}_{P}(A){\rm Var}_{P}(B_k) - {\rm Cov}^2_{P}(A,B_k)},
\end{align}
and for each $j$, we upper bound $\sup_{k \in \mathcal{K}_n}\big|\hat{\beta}_{jk} - \beta_k \big|$ as follows. 
\begin{align}
  &\sup_{k \in \mathcal{K}_n}\big|\hat{\beta}_{jk} - \beta_k \big|
  \le \sup_{k \in \mathcal{K}_n}\bigg|\bigg\{\frac{{\rm Var}_{\mathbb{P}_{j}}(A)}{{\rm Var}_{\mathbb{P}_j}(A){\rm Var}_{\mathbb{P}_j}(B_k) - {\rm Cov}^2_{\mathbb{P}_{j}}(A,B_k)}-\frac{{\rm Var}_P(A)}{{\rm Var}_P(A){\rm Var}_P(B_k) - {\rm Cov}^2_P(A,B_k)}\bigg\} \nonumber \\
  & \hspace{3.8cm} \times \bigg\{\frac{1}{j}\sum_{i=1}^j\big(B_{ik} - \mathbb{P}_j[B_k]\big)Y_i\bigg\}\bigg| \nonumber \\
  & \quad + \sup_{k \in \mathcal{K}_n}\bigg|\frac{{\rm Var}_P(A)}{{\rm Var}_P(A){\rm Var}_P(B_k) - {\rm Cov}^2_P(A,B_k)}\bigg\{\frac{1}{j}\sum_{i=1}^j\big(B_{ik} - \mathbb{P}_j[B_k]\big)Y_i - {\rm Cov}_P(B_{k},\tilde{Y})\bigg\}\bigg| \nonumber \\
  & \quad + \sup_{k \in \mathcal{K}_n}\bigg|\bigg\{\frac{{\rm Cov}_{\mathbb{P}_{j}}(A,B_k)}{{\rm Var}_{\mathbb{P}_j}(A){\rm Var}_{\mathbb{P}_j}(B_k) - {\rm Cov}^2_{\mathbb{P}_{j}}(A,B_k)} - \frac{{\rm Cov}_P(A,B_k)}{{\rm Var}_P(A){\rm Var}_P(B_k) - {\rm Cov}^2_P(A,B_k)}\bigg\}\nonumber \\
  & \hspace{1.8cm} \times\bigg\{\frac{1}{j}\sum_{i=1}^j\big(A_i - \mathbb{P}_j[A]\big)Y_i\bigg\}\bigg| \nonumber \\
  & \quad + \sup_{k \in \mathcal{K}_n}\bigg|\frac{{\rm Cov}_P(A,B_k)}{{\rm Var}_P(A){\rm Var}_P(B_k) - {\rm Cov}^2_P(A,B_k)}\bigg\{\frac{1}{j}\sum_{i=1}^j\big(A_i - \mathbb{P}_j[A]\big)Y_i-{\rm Cov}_{P}(A,\tilde{Y})\bigg\}\bigg| \label{eq:decomp.beta_k}
\end{align}
in which using the arguments for Lemmas S6.6-S6.7 of \cite{Huang2023}, the first and the third terms of the right-hand side can be shown $\lesssim K_{nj}$ given the occurrence of ${\cal A}_{n}$ and the boundedness of $\big|\,j^{-1}\sum_{i=1}^j\big(B_{ik} - \mathbb{P}_j[B_k]\big)Y_i\big|$ or $\big|\,j^{-1}\sum_{i=1}^j\big(A_i - \mathbb{P}_j[A]\big)Y_i\big|$
that follows \ref{assump:Mediator}, \ref{assump:Survival function}, the consistency of the Kaplan--Meier estimator and the boundedness of $A$ and $X$.

For the second term of the right-hand side on \eqref{eq:decomp.beta_k}, there exists a $(j,n)$-independent constant $K'$ so that
\begin{align}
  & \sup_{k \in \mathcal{K}_n}\bigg|\frac{1}{j}\sum_{i=1}^j\big(B_{ik} - \mathbb{P}_j[B_k]\big)Y_i - {\rm Cov}_P(B_{k},\tilde{Y})\bigg| \nonumber \\
  & = \; \sup_{k \in \mathcal{K}_n}\bigg|\frac{1}{j}\sum_{i=1}^j\big( \mathbb{P}_j[B_k] - P[B_k]\big)Y_i\bigg| + \sup_{k \in \mathcal{K}_n}\bigg|\frac{1}{j}\sum_{i=1}^j\big(B_{ik} - P[B_k]\big)\big(Y_i-\tilde{Y}_i\big)\bigg| \nonumber \\
  &\qquad + \sup_{k \in \mathcal{K}_n}\bigg|\frac{1}{j}\sum_{i=1}^j\big(B_{ik} - P[B_k]\big)\tilde{Y}_i - {\rm Cov}_P(B_{k},\tilde{Y})\bigg| \nonumber \\
  &\leq K'\bigg(K_{nj} +\sqrt{\frac{\log n}{n}}\bigg) + o_p(1),
\end{align}
following the occurrence of ${\cal A}_{n}$, \ref{assump:Mediator}, \ref{assump:Survival function}, the boundedness of $Y$ (that follows the consistency of the Kaplan--Meier estimator and the boundedness of $X$ along with \ref{assump:Survival function}), the weak law of large numbers, and
\begin{align*}
  &\sup_{k \in \mathcal{K}_n}\bigg|\frac{1}{j}\sum_{i=1}^j\big(B_{ik} - P[B_k]\big)\big(Y_i-\tilde{Y}_i\big)\bigg|
  \leq \sup_{k \in \mathcal{K}_n}\bigg|\bigg\{\frac{1}{j}\sum_{i=1}^j\big(B_{ik} - P[B_k]\big)^2\bigg\}^{1/2}\bigg\{\frac{1}{j}\sum_{i=1}^j\big(Y_i-\tilde{Y}_i\big)^2\bigg\}^{1/2}\bigg| \\
  & \quad \leq \sup_{k}\bigg|\sqrt{\sup_i\big(B_{ik} - P[B_k]\big)^2}\,\bigg|\sqrt{\sup_i\big(Y_i-\tilde{Y}_i\big)^2} \lesssim 
  \sqrt{\frac{\log n}{n}}
\end{align*}
implied by Cauchy--Schwarz inequality, \ref{assump:Mediator} and the occurrence of ${\cal B}_{n}$. 
Therefore, we can see that the second term of the right-hand side on \eqref{eq:decomp.beta_k} $\lesssim K_{nj}$, and the same for the last term of the right-hand side on \eqref{eq:decomp.beta_k}. Define the event \ref{eq:event_Cn_5} as ${\cal C}_{n5}$ and collecting all the results gives ${\rm P}(\mathcal{A}_n \cap \mathcal{B}_n \cap {\cal C}_{n5}^c) \to 0.$ 
}

Define the event \ref{eq:event_Cn_6} as ${\cal C}_{n6}$. As $\hat{Q}_{j}(u,k)$ is an exponential function of the maximum likelihood estimator (MLE) for the coefficient of $B_k$ in a univariate logistic regression for each $(u,k)$, the continuous mapping theorem implies that $\hat{Q}_{j}(u,k)$ is also a MLE. Therefore, ${\cal C}_{n6}$ is a direct consequence of \ref{assump:Mediator}, \ref{assump:Variances}, the continuous differentiability of exponential functions and the occurrence of $\mathcal{A}_n$, following the asymptotic behaviors of a M-estimator given by Theorem 5.39 and Example 5.40 in \cite{Vaart1998}.
Thus, we have ${\rm P}(\mathcal{A}_n \cap \mathcal{B}_n \cap {\cal C}_{n6}^c) \to 0.$

Letting $\mathcal{C}_n:=\cap_{q=1}^6\mathcal{C}_{nq}$, we have shown
${\rm P}(\mathcal{A}_n \cap \mathcal{B}_n \cap {\cal C}_n) \geq 1 - \sum_{q=1}^6{\rm P}(\mathcal{A}_n \cap \mathcal{B}_n \cap {\cal C}_{nq}^c) \to 1$.
\end{proof}

\begin{lemma}\label{lemma:Ln_prob_to_one}
Under the conditions of Theorem \ref{Thm:stab_one_step}, ${\rm P}(\mathcal{L}_n)\to 1$, where
$\mathcal{L}_n := \mathcal{A}_n\,\cap\, \mathcal{B}_n\,\cap\, \mathcal{C}_n\,\cap\, \mathcal{D}_n\,\cap\,\mathcal{E}_n\,\cap\,\mathcal{H}_n$.
\end{lemma}
\begin{proof}
The proof can be completed, following the arguments for Lemmas S6.9-S6.12 of \cite{Huang2023}, along with Lemmas \ref{lemma:An_Bn_prob_to_one} and \ref{lemma:An_Bn_Cn_prob_to_one}.
\end{proof}

\begin{lemma}\label{lemma:asymp.I}
Under the conditions of Theorem \ref{Thm:stab_one_step}, $\mbox{(I)}1_{\mathcal{L}_n}$ is asymptotically negligible.
\end{lemma}
\begin{proof}
By the definition of $\mbox{(I)}$ in \eqref{eq:def_terms} and using the triangle inequality,
\begin{align} \label{eq:decomp_I.1}
  & \big|\mbox{(I)}1_{\mathcal{L}_n}\big| \le
  \bigg|\frac{1_{\mathcal{L}_n}}{\sqrt{n-q_n}}\sum_{j=q_n}^{n-1}\frac{m_j}{\hat{\sigma}_{nj}}\bigg\{\hat{Q}_{j}(B_{j+1,\,k_j},k_j)\frac{A_{j+1}}{\mathbb{Q}_j(0)} - \frac{(1-A_{j+1})}{\mathbb{Q}_j(0)}\bigg\}\big[\hat{E}_j-E_0\big](1,B_{j+1,\,k_j},k_j)\bigg| \nonumber \\
  &+ \bigg|\frac{1_{\mathcal{L}_n}}{\sqrt{n-q_n}}\sum_{j=q_n}^{n-1}\frac{m_j}{\hat{\sigma}_{nj}}\frac{A_{j+1}}{\mathbb{Q}_j(1)}\bigg(1-\hat{Q}_{j}(B_{j+1,\,k_j},k_j)\frac{\mathbb{Q}_j(1)}{\mathbb{Q}_j(0)}\bigg) \int_{\mathcal{T}} \big[\hat{E}_j-E_0\big](A_{j+1},B_{j+1,\,k_j},s,k_j)\hat{M}_{j+1}({\rm d}s)\bigg|.
\end{align}
Let $M_{j+1}({\rm d}s) = 1(X_{j+1} \in {\rm d}s, \delta_{j+1}=0)-1(X_{j+1} \geq s)\Lambda({\rm d}s)$ and $\hat{M}_{j+1}({\rm d}s) - M_{j+1}({\rm d}s)
= -1(X_{j+1} \geq s)\big\{\hat{\Lambda}_n({\rm d}s)-\Lambda({\rm d}s)\big\}$, re-expressing the integral in the last term on the right-hand side of \eqref{eq:decomp_I.1}:
\begin{align*}
  &\big| 1_{{\cal L}_n}\int_{\mathcal{T}}\big[\hat{E}_j-E_0\big](A_{j+1},B_{j+1,\,k_j},s,k_j)\big\{\hat{M}_{j+1}({\rm d}s)-M_{j+1}({\rm d}s)\big\} \big| \\
  & = \big| 1_{{\cal L}_n}\int_{\mathcal{T}}\big[\hat{E}_j-E_0\big](A_{j+1},B_{j+1,\,k_j},s,k_j)1(X_{j+1} \geq s) \big\{\hat{\Lambda}_n({\rm d}s)-\Lambda({\rm d}s)\big\} \big|.
\end{align*}
We further apply Lemma S6.13 of \cite{Huang2023} to \eqref{eq:decomp_I.1}, taking
\begin{align} \label{eq:ub_hatE_j}
  f_{n1}(O_{j+1,k_j}) = \big[\hat{E}_j-E_0\big](1,B_{j+1,\,k_j},k_j)\mbox{ with }
  \sup_{o \in \mathcal{X}}\big|f_{n1}(o)1_{{\cal L}_n}\big| \leq \sup_{(j,k,u)}\big|\big[\hat{E}_j-E_0\big](1,u,k)\big| \lesssim K_{nq_n},
\end{align}
and
\begin{align*}
  f_{n2}(O_{j+1,k_j}) = \int_{\mathcal{T}}\big[\hat{E}_j-E_0\big](A_{j+1},B_{j+1,\,k_j},s,k_j)1(X_{j+1} \geq s)\big\{\hat{\Lambda}_n({\rm d}s)-\Lambda({\rm d}s)\big\}
\end{align*}
with $\sup_{o \in \mathcal{X}}\big|f_{n2}(o)1_{{\cal L}_n}\big| \lesssim \sqrt{\log(n \lor p_n)}/q_n^{1/4}$,
where the upper bound of $\sup_{o \in \mathcal{X}}\big|f_{n2}(o)1_{{\cal L}_n}\big|$ is derived using (S6.15.2) of Lemma S6.15 in \cite{Huang2023}.

Continuing \eqref{eq:decomp_I.1},
we further bound $\big|\mbox{(I)}1_{\mathcal{L}_n}\big|$ by
\begin{align} \label{eq:decomp_I}
  & \big|\mbox{(I)}1_{\mathcal{L}_n}\big| \le
  \bigg|\frac{1_{\mathcal{L}_n}}{\sqrt{n-q_n}}\sum_{j=q_n}^{n-1}\frac{m_j}{\sigma_{nj}}\frac{1}{Q(0)}\bigg\{Q(B_{j+1,\,k_j},k_j)A_{j+1} - (1-A_{j+1})\bigg\}\big[\hat{E}_j-E_0\big](1,B_{j+1,\,k_j},k_j)\bigg| \nonumber \\
  &\quad + \bigg|\frac{1_{\mathcal{L}_n}}{\sqrt{n-q_n}}\sum_{j=q_n}^{n-1}\frac{m_j}{\sigma_{nj}}\frac{A_{j+1}}{Q(1)}\bigg(1-Q(B_{j+1,\,k_j},k_j)\frac{Q(1)}{Q(0)}\bigg) \nonumber \\
  & \hspace{3cm} \times \int_{\mathcal{T}} \big[\hat{E}_j-E_0\big](A_{j+1},B_{j+1,\,k_j},s,k_j)1(X_{j+1} \ge s)\big\{\hat{\Lambda}_n({\rm d}s)-\Lambda({\rm d}s)\big\}\bigg| + o_p(1),
\end{align}
where the penultimate term can be shown to converge to zero in probability by using the arguments for Lemma S6.18 of \cite{Huang2023}.

Now we deal with the first term on the right-hand-side of \eqref{eq:decomp_I}. Fix $n$ and for $j \in \{q_n+1,\ldots, n\}$,
\begin{align*}
  \bar{H}_{nj} := \frac{1}{\sqrt{n-q_n}}\frac{m_{j-1}}{\sigma_{n{j-1}}}\bigg\{Q(B_{j,\,k_{j-1}},k_{j-1})\frac{A_{j}}{Q(0)} - \frac{(1-A_{j})}{Q(0)}\bigg\}\big[\hat{E}_{j-1}-E_0\big](1,B_{j,\,k_{j-1}},k_{j-1}).
\end{align*}
By \eqref{eq:ub_hatE_j}, along with the uniform boundedness of $B_k$ given in \ref{assump:Mediator} and that $\sigma_{nj}$ is uniformly bounded away from zero as assumed in \ref{assump:Variances},
there exists a $\bar{U}_n := K'\sqrt{\log(n \lor p_n)}\,\big/\sqrt{(n-q_n)q_n}$ for some constant $K'>0$ such that $\max_j\big|\bar{H}_{nj}\big| \le \bar{U}_n$. Define the filtration ${\cal O}_{nj} = \sigma(\{O_1,\ldots,O_j\}, \bs{B}_{j+1})$.
We know that $\{(\bar{H}_{nj}, {\cal O}_{nj}),\, j=q_n+1,\ldots, n\}$ is a martingale difference sequence because $E\big|\bar{H}_{nj}\big| < \infty$; $\bar{H}_{nj}$ is ${\cal O}_{nj}$-measurable, and for $j=q_n,\ldots, n-1$,
\begin{align*}  
  & E\big[\bar{H}_{n,\,j+1}\big|{\cal O}_{nj}\big]\\
  & = \frac{1}{\sqrt{n-q_n}}\frac{m_{j}}{\sigma_{nj}Q(0)}\big[\hat{E}_{j}-E_0\big](1,B_{j+1,\,k_j},k_j) \Big\{Q(B_{j+1,\,k_j},k_j)E\big[A_{j+1}\big|{\cal O}_{nj}\big] - E\big[1-A_{j+1}\big|{\cal O}_{nj}\big]\Big\} = 0
\end{align*}
because the value in the curly brackets is precisely zero. Then for $\varepsilon > 0$,
\begin{align*}
  & {\rm P}\bigg(\,\bigg|\sum_{j=q_n}^{n-1}\bar{H}_{n,\,j+1}\bigg| \ge \varepsilon \bigg) \le \varepsilon^{-2}E\bigg[\bigg(\sum_{j=q_n}^{n-1}\bar{H}_{n,\,j+1}\bigg)^2\,\bigg] \\
  & = \varepsilon^{-2}\bigg(\sum_{j=q_n}^{n-1}E\big[\bar{H}_{n,\,j+1}^2\big] + 2\sum_{q_n\le i< j\le n-1}E\Big[\bar{H}_{n,\,i+1}E\Big[\bar{H}_{n,\,j+1} \big| \mathcal{O}_{nj}\Big]\Big]\bigg) \\
  & = \varepsilon^{-2}\sum_{j=q_n}^{n-1}E\big[ \bar{H}_{n,\,j+1}^2 \big]
  \le \varepsilon^{-2}(n-q_n)\bar{U}_n^2 \to 0.
\end{align*}
As $\big|\mbox{(I)}1_{{\cal L}_n}\big| \lesssim \big|\sum_{j=q_n}^{n-1}\bar{H}_{n,\,j+1}\big| + o_p(1)$, we conclude that $\mbox{(I)}1_{{\cal L}_n}=o_p(1)$.
\end{proof}

\begin{lemma}\label{lemma:asymp.II}
Under the conditions of Theorem \ref{Thm:stab_one_step}, $\mbox{(II)}1_{\mathcal{L}_n}$ is asymptotically negligible.
\end{lemma}
\begin{proof}
By the definition of $\mbox{(II)}$ in \eqref{eq:def_terms} and appealing to Lemma S6.13 of \cite{Huang2023} with
\begin{align*}
  f_n(O_{j+1,k_j}) = \int_{\mathcal{T}} E_0(A_{j+1},B_{j+1,\,k_j},s,k_j)1(X_{j+1} \geq s)[\hat{\Lambda}_n-\Lambda]({\rm d}s),  
\end{align*}
where $\sup_{o \in \mathcal{X}}\big|\,f_n(o)1_{{\cal L}_n}\big| \lesssim \sqrt{\log(n \lor p_n)}\,/q_n^{1/4}$ by following the steps in Lemma S6.15 of \cite{Huang2023}, we have the second term of $\mbox{(II)}1_{\mathcal{L}_n}$ bounded by
\begin{align} \label{eq:expansion}
  \bigg|\frac{1}{\sqrt{n-q_n}}\sum_{j=q_n}^{n-1}\frac{m_j}{\sigma_{nj}}\frac{A_{j+1}}{Q(1)}&\bigg(1-Q(B_{j+1,\,k_j},k_j)\frac{Q(1)}{Q(0)}\bigg) \nonumber \\
  & \times \int_{\mathcal{T}} E_0(A_{j+1},B_{j+1,\,k_j},s,k_j)1(X_{j+1} \ge s)[\hat{\Lambda}_n-\Lambda]({\rm d}s)\bigg|1_{\mathcal{L}_n} + o_p(1).
\end{align}
The above display can be shown of order $o_p(1)$, following Lemmas S6.16--S6.17 of \cite{Huang2023}.

First-order Taylor expanding the first term in $\mbox{(II)}1_{\mathcal{L}_n}$ around $G$ yields the approximation below
\begin{align} \label{eq:expansion.1}
  & \frac{1}{\sqrt{n-q_n}}\sum_{j=q_n}^{n-1}\frac{m_j}{\hat{\sigma}_{nj}}\bigg\{\frac{1}{\mathbb{Q}_j(1)} - 
  \frac{\hat{Q}_j(B_{j+1,\,k_j},k_j)}{\mathbb{Q}_j(0)}\bigg\}A_{j+1}\delta_{j+1}X_{j+1}\Big(\frac{1}{\hat{G}_n(X_{j+1})}-\frac{1}{G(X_{j+1})}\Big)1_{\mathcal{L}_n} \nonumber \\
  & = -\frac{1}{\sqrt{n-q_n}}\sum_{j=q_n}^{n-1}\frac{m_j}{\hat{\sigma}_{nj}}\bigg\{\frac{1}{\mathbb{Q}_j(1)} - 
  \frac{\hat{Q}_j(B_{j+1,\,k_j},k_j)}{\mathbb{Q}_j(0)}\bigg\}A_{j+1}\tilde{Y}_{j+1}\Big(\frac{\hat{G}_n(X_{j+1})}{G(X_{j+1})} - 1\Big)1_{\mathcal{L}_n} + o_p(n^{-1/2}) \nonumber \\ 
  & = -\frac{1}{\sqrt{n-q_n}}\sum_{j=q_n}^{n-1}\frac{m_j}{\sigma_{nj}}\bigg\{\frac{1}{Q(1)} - 
  \frac{Q(B_{j+1,\,k_j},k_j)}{Q(0)}\bigg\}A_{j+1}\tilde{Y}_{j+1}\Big(\frac{\hat{G}_n(X_{j+1})}{G(X_{j+1})} - 1\Big)1_{\mathcal{L}_n} \nonumber\\
  &\quad - \frac{1}{\sqrt{n-q_n}}\sum_{j=q_n}^{n-1}\bigg\{\frac{m_j}{\hat{\sigma}_{nj}}\bigg(\frac{1}{\mathbb{Q}_j(1)} - 
  \frac{\hat{Q}_j(B_{j+1,\,k_j},k_j)}{\mathbb{Q}_j(0)}\bigg) - \frac{m_j}{\sigma_{nj}}\bigg(\frac{1}{Q(1)} - 
  \frac{Q(B_{j+1,\,k_j},k_j)}{Q(0)}\bigg)\bigg\}A_{j+1}\tilde{Y}_{j+1} \nonumber \\
  & \hspace{3cm} \times \Big(\frac{\hat{G}_n(X_{j+1})}{G(X_{j+1})} - 1\Big)1_{\mathcal{L}_n}
  + o_p(n^{-1/2}).
\end{align}
The middle term in the last line of \eqref{eq:expansion.1} is shown to be $o_p(1)$ as follows. Using Cauchy--Schwarz inequality,
\begin{align}\label{eq:expansion.2}
  &\bigg|\frac{1}{\sqrt{n-q_n}}\sum_{j=q_n}^{n-1}\bigg\{\frac{m_j}{\hat{\sigma}_{nj}}\bigg(\frac{1}{\mathbb{Q}_j(1)} - 
  \frac{\hat{Q}_j(B_{j+1,\,k_j},k_j)}{\mathbb{Q}_j(0)}\bigg) - \frac{m_j}{\sigma_{nj}}\bigg(\frac{1}{Q(1)} - 
  \frac{Q(B_{j+1,\,k_j},k_j)}{Q(0)}\bigg)\bigg\}A_{j+1}\tilde{Y}_{j+1} \nonumber \\
  & \hspace{3cm} \times \Big(\frac{\hat{G}_n(X_{j+1})}{G(X_{j+1})} - 1\Big)\bigg|1_{\mathcal{L}_n} \nonumber\\
  & \le \bigg\{\frac{1}{n-q_n}\sum_{j=q_n}^{n-1}\bigg\{\frac{m_j}{\hat{\sigma}_{nj}}\bigg(\frac{1}{\mathbb{Q}_j(1)} - 
  \frac{\hat{Q}_j(B_{j+1,\,k_j},k_j)}{\mathbb{Q}_j(0)}\bigg) - \frac{m_j}{\sigma_{nj}}\bigg(\frac{1}{Q(1)} - 
  \frac{Q(B_{j+1,\,k_j},k_j)}{Q(0)}\bigg)\bigg\}^2A_{j+1}^2\tilde{Y}_{j+1}^2\bigg\}^{1/2} \nonumber \\
  & \qquad \times \bigg\{\frac{1}{n-q_n}\sum_{j=q_n}^{n-1}\Big(\frac{\hat{G}_n(X_{j+1})}{G(X_{j+1})} - 1\Big)^2\bigg\}^{1/2}1_{\mathcal{L}_n}. 
\end{align}
With the occurrence of $\mathcal{B}_n$ contained in $\mathcal{L}_n$, the right-hand side of \eqref{eq:expansion.2} is upper-bounded by
\begin{align}\label{eq:ub.expansion.2}
  &\bigg\{\frac{1}{n-q_n}\sum_{j=q_n}^{n-1}\bigg\{\frac{m_j}{\hat{\sigma}_{nj}}\bigg(\frac{1}{\mathbb{Q}_j(1)} - 
  \frac{\hat{Q}_j(B_{j+1,\,k_j},k_j)}{\mathbb{Q}_j(0)}\bigg) - \frac{m_j}{\sigma_{nj}}\bigg(\frac{1}{Q(1)} - 
  \frac{Q(B_{j+1,\,k_j},k_j)}{Q(0)}\bigg)\bigg\}^2A_{j+1}^2\tilde{Y}_{j+1}^2\bigg\}^{1/2}\sqrt{\frac{\log n}{n}}1_{\mathcal{L}_n} \nonumber \\
  & \lesssim \bigg\{\frac{1}{n-q_n}\sum_{j=q_n}^{n-1}\bigg\{\bigg(\frac{1}{\mathbb{Q}_j(1)} - 
  \frac{\hat{Q}_j(B_{j+1,\,k_j},k_j)}{\mathbb{Q}_j(0)}\bigg) -\bigg(\frac{1}{Q(1)} - 
  \frac{Q(B_{j+1,\,k_j},k_j)}{Q(0)}\bigg)\bigg\}^2\bigg\}^{1/2}\sqrt{\frac{\log n}{n}}1_{\mathcal{L}_n},
\end{align}
where the second line follows that $\min_j\hat{\sigma}_{nj}^2$ is bounded away from zero by the occurrence of $\mathcal{E}_n$ contained in $\mathcal{L}_n$, $\sigma_{nj}^2 = {\rm Var}(\,f^*_{k_j}(O\,|\,P))$ is bounded away from zero by \ref{assump:Variances} for each $j$, $m_j \in \{-1,1\}$, $A_{j+1} \in \{0,1\}$ and \ref{assump:Survival function} assumes $G(\tau) >0$ that implies $\tilde{Y}_{j + 1} = \delta_{j+1}X_{j+1}/G(X_{j+1})$ is uniformly bounded on $j$. 
Along with the decomposition
\begin{align*}
  &\bigg(\frac{1}{\mathbb{Q}_j(1)} - 
  \frac{\hat{Q}_j(B_{j+1,\,k_j},k_j)}{\mathbb{Q}_j(0)}\bigg) -\bigg(\frac{1}{Q(1)} - 
  \frac{Q(B_{j+1,\,k_j},k_j)}{Q(0)}\bigg) \\
  & = \bigg(\frac{1}{\mathbb{Q}_j(1)} - \frac{1}{Q(1)}\bigg) + \Big(\hat{Q}_j(B_{j+1,\,k_j},k_j) - Q(B_{j+1,\,k_j},k_j)\Big)\bigg(\frac{-1}{\mathbb{Q}_j(0)}\bigg)\\
  & \qquad + \bigg(\frac{1}{\mathbb{Q}_j(0)} - \frac{1}{Q(0)}\bigg)\Big(-Q(B_{j+1,\,k_j},k_j)\Big),
\end{align*}
the upper bound of the right-hand side in \eqref{eq:ub.expansion.2} is given as
\begin{align}\label{eq:ub.expansion.3}
  & \frac{\sqrt{3}}{\sqrt{n-q_n}}\bigg\{\sum_{j=q_n}^{n-1}\bigg(\frac{1}{\mathbb{Q}_j(1)} - \frac{1}{Q(1)}\bigg)^2 + 
  \sum_{j=q_n}^{n-1}\frac{\Big(\hat{Q}_j(B_{j+1,\,k_j},k_j) - Q(B_{j+1,\,k_j},k_j)\Big)^2}{\mathbb{Q}^2_j(0)} \nonumber \\
  & \hspace{1.6cm} + \sum_{j=q_n}^{n-1}
  Q^2(B_{j+1,\,k_j},k_j)\bigg(\frac{1}{\mathbb{Q}_j(0)} - \frac{1}{Q(0)}\bigg)^2\bigg\}^{1/2}\sqrt{\frac{\log n}{n}}1_{\mathcal{L}_n} \nonumber \\
  & \lesssim \sqrt{n-q_n}\bigg[\max_{j \in \{q_n,\ldots, n-1\}}\Big\{2/j + K_{nj}^2 \Big\}\bigg]\sqrt{\frac{\log n}{n}}1_{\mathcal{L}_n} = \sqrt{n-q_n}\bigg\{\frac{2}{q_n} + \frac{\log(n \lor p_n)}{q_n} \bigg\}\sqrt{\frac{\log n}{n}}1_{\mathcal{L}_n} \to 0,
\end{align}
where the first inequality follows from \ref{assump:Positivity} and the occurrence of events \ref{eq:event_Cn_2}, \ref{eq:event_Cn_3}, \ref{eq:event_Cn_6} of $\mathcal{C}_n$ contained in $\mathcal{L}_n$, and the convergence to zero results from the conditions $n - q_n \to \infty$, $n\big/q_n = O(1)$ as well as $q_n^{1/4}\big/log(n \lor p_n) \to \infty$. We have shown the middle term in \eqref{eq:expansion.1} is $o_p(1)$.

To show $\mbox{(II)}1_{\mathcal{L}_n}$ is asymptotically negligible, it remains to address the first term in the last equality of \eqref{eq:expansion.1}. We can upper bound it using Cauchy--Schwarz inequality and the occurrence of $\mathcal{B}_n$ contained in $\mathcal{L}_n$: 
\begin{align} \label{eq:decomposition}
  &\big|\mbox{(II)}\big|1_{\mathcal{L}_n} = \bigg|\frac{1}{\sqrt{n-q_n}}\sum_{j=q_n}^{n-1}\frac{m_j}{\sigma_{nj}}\bigg\{\frac{1}{Q(1)} - 
  \frac{Q(B_{j+1,\,k_j},k_j)}{Q(0)}\bigg\}A_{j+1}\tilde{Y}_{j+1}\Big(\frac{\hat{G}_n(X_{j+1})}{G(X_{j+1})} - 1\Big)\bigg|1_{\mathcal{L}_n} + o_p(n^{-1/2}) \nonumber \\
  & \leq \bigg\{\frac{1}{n-q_n}\sum_{j=q_n}^{n-1}\frac{1}{\sigma_{nj}^2}\bigg\{\frac{1}{Q(1)} - 
  \frac{Q(B_{j+1,\,k_j},k_j)}{Q(0)}\bigg\}^2A_{j+1}^2\tilde{Y}_{j+1}^2\bigg\}^{1/2}\bigg\{\frac{1}{n-q_n}\sum_{j=q_n}^{n-1}\bigg(\frac{\hat{G}_n(X_{j+1})}{G(X_{j+1})} - 1\bigg)^2\bigg\}^{1/2}1_{\mathcal{L}_n}\nonumber \\ 
  & \qquad + o_p(n^{-1/2}) \nonumber \\
  & \leq \bigg\{\frac{1}{n-q_n}\sum_{j=q_n}^{n-1}\frac{1}{\sigma_{nj}^2}\bigg\{\frac{1}{Q(1)} - 
  \frac{Q(B_{j+1,\,k_j},k_j)}{Q(0)}\bigg\}^2A_{j+1}^2\tilde{Y}_{j+1}^2\bigg\}^{1/2}\sqrt{\frac{\log n}{n}}1_{\mathcal{L}_n} + o_p(n^{-1/2}).
\end{align}
Moreover, \ref{assump:Positivity} and \ref{assump:Mediator} give that $\max_{j}\big\{1/Q(1) - Q(B_{j+1,\,k_j},k_j)\big/Q(0)\big\}^2$ is bounded above, \ref{assump:Survival function} assumes $G(\tau) >0$ that implies $\tilde{Y}_{j + 1} = \delta_{j+1}X_{j+1}/G(X_{j+1})$ is uniformly bounded on $j$, and we see $\sigma_{nj}^2 = {\rm Var}(\,f^*_{k_j}(O\,|\,P))$ be bounded away from zero by \ref{assump:Variances}. Along with $A_{j+1} \in \{0,1\}$ for each $j$, there exists a $(j,n)$-independent constant $K' \in (0, \infty)$ such that
\begin{align*}
  & \frac{1}{n-q_n}\sum_{j=q_n}^{n-1}\frac{1}{\sigma_{nj}^2}\bigg\{\frac{1}{Q(1)} - \frac{Q(B_{j+1,\,k_j},k_j)}{Q(0)}\bigg\}^2A_{j+1}^2\tilde{Y}_{j+1}^2 \\
  & \le \max_{j \in \{q_n,\ldots, n-1\}}\bigg\{\frac{1}{\sigma_{nj}^2}\bigg(\frac{1}{Q(1)} - \frac{Q(B_{j+1,\,k_j},k_j)}{Q(0)}\bigg)^2A_{j+1}^2\tilde{Y}_{j+1}^2\bigg\} \le K'.
\end{align*}
From \eqref{eq:decomposition}, together with the results in \eqref{eq:expansion}-\eqref{eq:ub.expansion.3}, we conclude that $\mbox{(II)}1_{\mathcal{L}_n}$ is $o_p(n^{-1/2})$.
\end{proof}

\begin{lemma}\label{lemma:asymp.III}
Under the conditions of Theorem \ref{Thm:stab_one_step}, $\mbox{(III)}1_{\mathcal{L}_n}$ is asymptotically negligible.
\end{lemma}
\begin{proof}
As observing $\int_{\mathcal{T}} \tilde{f}(s)M_{j+1}({\rm d}s) = 0$ for any predictable function $\tilde{f}$ with respect to the filtration ${\cal O}_{nj} := \sigma(\{O_1,\ldots,O_j\}, A_{j+1}, \bs{B}_{j+1})$, along with the properties given the occurrence of $\mathcal{L}_n$,
we can discard the last two terms in the decomposition of $\mbox{(III)}$ in \eqref{eq:def_terms}.
Thus, it remains to show $\big|\mbox{(III)}1_{{\cal L}_n}\big| = o_p(1)$, where
\begin{align}\label{eq:decomp_III}
  &\big|\mbox{(III)}1_{{\cal L}_n}\big| \le \bigg|\frac{1_{{\cal L}_n}}{\sqrt{n-q_n}}\sum_{j=q_n}^{n-1}\frac{m_j}{\hat{\sigma}_{nj}}\bigg(\frac{1}{\mathbb{Q}_j(0)}-\frac{1}{Q(0)}\bigg)(1-A_{j+1})\Big[E_0(1,B_{j+1,\,k_j},k_j) - {\color {change} \hat{\beta}_{jk_j}\hat{q}_j(0,k_j) }\Big] \bigg| \nonumber \\
  &\hspace{1.8cm} + \bigg|\frac{1_{{\cal L}_n}}{\sqrt{n-q_n}}\sum_{j=q_n}^{n-1}\frac{m_j}{\hat{\sigma}_{nj}}\bigg(\frac{1}{\mathbb{Q}_j(1)} - \frac{1}{Q(1)}\bigg)A_{j+1}\Big[\tilde{Y}_{j+1} - {\color {change} \hat{\beta}_{jk_j}\hat{q}_j(1,k_j)} \Big]\bigg| \nonumber \\
  &\hspace{1.8cm} + \bigg|\frac{1_{{\cal L}_n}}{\sqrt{n-q_n}}\sum_{j=q_n}^{n-1}\frac{m_j}{\hat{\sigma}_{nj}}\bigg(1-\frac{(1-A_{j+1})}{Q(0)}\bigg){\color{change} \Big[\hat{\beta}_{jk_j}\hat{q}_j(0,k_j)-\beta_{k_j}q(0,k_j)\Big]}\bigg| \nonumber \\
  &\hspace{1.8cm} + \bigg|\frac{1_{{\cal L}_n}}{\sqrt{n-q_n}}\sum_{j=q_n}^{n-1}\frac{m_j}{\hat{\sigma}_{nj}}{\color {change} \bigg(1-\frac{A_{j+1}}{Q(1)}\bigg)\Big[\hat{\beta}_{jk_j}\hat{q}_j(1,k_j)-\beta_{k_j}q(1,k_j)\Big]}\bigg| \\
  &\hspace{1.8cm} + \bigg|\frac{1_{{\cal L}_n}}{\sqrt{n-q_n}}\sum_{j=q_n}^{n-1}\frac{m_j}{\hat{\sigma}_{nj}} A_{j+1}
  \Big[\tilde{Y}_{j+1}-E_0(1,B_{j+1,\,k_j},k_j)\Big] \nonumber \\
  &\hspace{4cm} \times \bigg\{\frac{1}{\mathbb{Q}_j(0)}\big[\hat{Q}_j-Q\big](B_{j+1,\,k_j},k_j) + \bigg(\frac{1}{\mathbb{Q}_j(0)}-\frac{1}{Q(0)}\bigg)Q(B_{j+1,\,k_j},k_j)\bigg\}
   \bigg|. \nonumber
\end{align}

Note that for $a \in \{0,1\}$ and $j \in \{q_n,\ldots,n\}$,
\begin{align} \label{eq:bounded.diff.betaQ}
  \sup_{k \in \mathcal{K}_n}\big|\hat{\beta}_{jk}\hat{q}_j(a,k)-\beta_{k}q(a,k)\big| \le \sup_{k \in \mathcal{K}_n}\big|\big(\hat{\beta}_{jk}-\beta_{k}\big)\hat{q}_j(a,k)\big| + \sup_{k \in \mathcal{K}_n}\big|\big(\hat{q}_j(a,k)-q(a,k)\big)\beta_{k}\big| \lesssim K_{nj},
\end{align}
following \ref{assump:Mediator} and the occurrence of events \ref{eq:event_Cn_4}--\ref{eq:event_Cn_5} of $\mathcal{C}_n$ contained in $\mathcal{L}_n$, which further implies that
\begin{align} \label{eq:bounded.betaQ}
  \sup_{k \in \mathcal{K}_n}\big|\hat{\beta}_{jk}\hat{q}_j(a,k)\big| \lesssim K_{nj} + \sup_{k \in \mathcal{K}_n}\big|\beta_{k}q(a,k)\big|.
\end{align}
The first term in the decomposition \eqref{eq:decomp_III} is shown of order $o_p(1)$ as follows. 
Let ${\cal O}_{nj} := \sigma(\{O_1,\ldots,O_j\})$ and 
\begin{align*}
  \tilde{H}_{nj} := \frac{1}{\sqrt{n-q_n}}\frac{m_{j-1}}{\hat{\sigma}_{nj-1}}\bigg(\frac{1}{\mathbb{Q}_{j-1}(0)}-\frac{1}{Q(0)}\bigg)(1-A_{j})\Big[E_0(1,B_{j,\,k_{j-1}},k_{j-1}) - {\color {change} \hat{\beta}_{j-1k_{j-1}}\hat{q}_{j-1}(0,k_{j-1})}\Big].
\end{align*}
Then the first term in the decomposition \eqref{eq:decomp_III} is equal to $\big\{\sum_{j=q_n}^{n-1}\tilde{H}_{n,\,j+1}\big\}1_{\mathcal{L}_n}$.
Moreover, we see that $E\big|\tilde{H}_{nj}\big| < \infty$, following 
that $\min_{j}\hat{\sigma}^2_{nj}$ is bounded away from zero by the occurrence of $\mathcal{E}_n$, the occurrence of event \ref{eq:event_Cn_2}--\ref{eq:event_Cn_3} of $\mathcal{C}_n$ contained in $\mathcal{L}_n$,
\ref{assump:Positivity} that implies $Q(a)$ is bounded away from zero for $a \in \mathcal{A}$,
\ref{assump:Conditional_mean_E0} that upper-bounds $\sup_{(a,k,u)}\big|E_0(a, u, k)\big|$, and the boundedness of $\hat{\beta}_{j-1k_{j-1}}\hat{q}_{j-1}(0,k_{j-1})$ given in \eqref{eq:bounded.betaQ}.
Moreover, along with the occurrence of event \ref{eq:event_Cn_3} of $\mathcal{C}_n$ contained in $\mathcal{L}_n$, the result in \eqref{eq:bounded.betaQ} yields 
$\big|\tilde{H}_{nj}\big| \lesssim \sqrt{\log(n \lor p_n)\,/\,[q_n(n-q_n)]} \equiv \tilde{U}_n$, for each $j$.
Note that
\begin{align*}
  & E\big[\tilde{H}_{n,\,j+1}\big|{\cal O}_{nj}\big] \\
  & = \frac{1}{\sqrt{n-q_n}}\frac{m_j}{\hat{\sigma}_{nj}}\bigg(\frac{1}{\mathbb{Q}_j(0)}-\frac{1}{Q(0)}\bigg)E\Big[&(1-A_{j+1})\Big(E_0(1,B_{j+1,\,k_j},k_j) - {\color {change} \hat{\beta}_{jk_j}\hat{q}_j(0,k_j)}\Big)\big|{\cal O}_{nj}\Big]=0,
\end{align*}
resulting from
\begin{align*}
 & E\Big[(1-A_{j+1})\Big(E_0(1,B_{j+1,\,k_j},k_j) - {\color {change} \hat{\beta}_{jk_j}\hat{q}_j(0,k_j)}\Big)\big|{\cal O}_{nj}\Big] \\
 & = P(A_{j+1}=0)\Big\{E\big[E_0(1,B_{j+1,\,k_j},k_j)\big|A_{j+1}=0, {\cal O}_{nj}\big] - {\color {change} \hat{\beta}_{jk_j}\hat{q}_j(0,k_j)}\Big\} = 0.
\end{align*}
As $\tilde{H}_{nj}$ is seen ${\cal O}_{nj}$-measurable,
$\{(\tilde{H}_{nj}, {\cal O}_{nj}),\, j=q_n+1,\ldots,n\}$ is a martingale difference sequence.
Then Chebyshev's inequality implies that for $\varepsilon > 0$,
\begin{align*}
  & {\rm P}\bigg(\,\bigg|\sum_{j=q_n}^{n-1}\tilde{H}_{n,\,j+1}\bigg| \ge \varepsilon \bigg) \le \varepsilon^{-2}E\bigg[\bigg(\sum_{j=q_n}^{n-1}\tilde{H}_{n,\,j+1}\bigg)^2\,\bigg] \\
  & = \varepsilon^{-2}\bigg(\sum_{j=q_n}^{n-1}E\big[\tilde{H}_{n,\,j+1}\big]^2 + 2\sum_{q_n\le i< j\le n-1}E\Big[\tilde{H}_{n,\,i+1}E\Big[\tilde{H}_{n,\,j+1} \big| \mathcal{O}_{nj}\Big]\Big]\bigg) \\
  & = \varepsilon^{-2}\sum_{j=q_n}^{n-1}E\big[ \tilde{H}_{n,\,j+1}\big]^2 \le \varepsilon^{-2}(n-q_n)\tilde{U}_n^2 \to 0;
\end{align*}
this result disposes of the first term in the decomposition \eqref{eq:decomp_III}.

Similar arguments can be applied to show the second term in the decomposition \eqref{eq:decomp_III} of order $o_p(1)$, taking
\begin{align*}
  &\tilde{H}_{nj} := \frac{1}{\sqrt{n-q_n}}\frac{m_{j-1}}{\hat{\sigma}_{nj-1}}\bigg(\frac{1}{\mathbb{Q}_{j-1}(1)}-\frac{1}{Q(1)}\bigg)A_{j}\Big[\tilde{Y}_{j} - {\color {change} \hat{\beta}_{j-1k_{j-1}}\hat{q}_{j-1}(1,k_{j-1})}\Big]
\end{align*}
and following
$E\big[\tilde{Y}_{j+1}\big|A_{j+1}=1,\mathcal{O}_{nj}\big] = \hat{\beta}_{jk_j}\hat{q}_j(1,k_j)$.

We can handle the third term in the decomposition \eqref{eq:decomp_III} in a parallel fashion.
Provided ${\cal O}_{nj} = \sigma(\{O_1,\ldots,O_j\})$, let 
\begin{align*}
 H'_{nj} := \frac{1}{\sqrt{n-q_n}}\frac{m_{\,j\,-1}}{\hat{\sigma}_{n\,j\,-1}}\bigg(1-\frac{(1-A_{j})}{Q(0)}\bigg)\Big[\hat{\beta}_{j-1k_{j-1}}\hat{q}_{j-1}(0,k_{j-1})-\beta_{k_{j-1}}q(0,k_{j-1})\Big];
\end{align*}
the third term in the decomposition \eqref{eq:decomp_III} is equal to $\big\{\sum_{j=q_n}^{n-1}H'_{n,\,j+1}\big\}1_{\mathcal{L}_n}$.
It is easy to see that $E\big|H'_{nj}\big| < \infty$, following that both $\min_{j}\hat{\sigma}^2_{nj}$ and $Q(a)$ for $a \in \mathcal{A}$ are bounded away from zero respectively by the occurrence of $\mathcal{E}_n$ and \ref{assump:Positivity}, 
that $\big|\hat{\beta}_{jk_j}\hat{q}_j(0,k_j)-\beta_{k_j}q(0,k_j)\big|$ is upper-bounded as given in \eqref{eq:bounded.diff.betaQ}.
Moreover for each $j$, 
$\big|H'_{nj}\big| \lesssim \sqrt{\log(n \lor p_n)\,/\,[q_n(n-q_n)]} \equiv U'_n$.
As $H'_{nj}$ is seen ${\cal O}_{nj}$-measurable, along with
\begin{align*}
  E\big[H'_{n,\,j+1}\big|{\cal O}_{nj}\big] = \frac{1}{\sqrt{n-q_n}}\frac{m_j}{\hat{\sigma}_{nj}} E\bigg[1-\frac{(1-A_{j+1})}{Q(0)}\bigg| {\cal O}_{nj}\bigg] \Big[\hat{\beta}_{jk_j}\hat{q}_{j}(0,k_j)-\beta_{k_j}q(0,k_j)\Big] = 0,
\end{align*}
$\{(H'_{nj}, {\cal O}_{nj}),\, j=q_n+1,\ldots,n\}$ is a martingale difference sequence.
Similarly, Chebyshev's inequality implies that
for $\varepsilon > 0$, 
\begin{align*}
  {\rm P}\bigg(\,\bigg|\sum_{j=q_n}^{n-1}H'_{n,\,j+1}\bigg| \ge \varepsilon \bigg) \le \varepsilon^{-2}E\bigg[\bigg(\sum_{j=q_n}^{n-1}H'_{n,\,j+1}\bigg)^2\,\bigg]
  = \varepsilon^{-2}\sum_{j=q_n}^{n-1}E\big[ H_{n,\,j+1}'\big]^2 \le \varepsilon^{-2}(n-q_n)(U_n')^2 \to 0; 
\end{align*}
this result gives the asymptotic ineligibility of the third term in the decomposition \eqref{eq:decomp_III}.

Similar arguments can be used to show the fourth term in the decomposition \eqref{eq:decomp_III} of order $o_p(1)$, taking
\begin{align*}
 H'_{nj} := \frac{1}{\sqrt{n-q_n}}\frac{m_{\,j\,-1}}{\hat{\sigma}_{n\,j\,-1}}\bigg(1-\frac{A_{j}}{Q(1)}\bigg)\Big[\hat{\beta}_{j-1k_{j-1}}\hat{q}_{j-1}(1,k_{j-1})-\beta_{k_{j-1}}q(1,k_{j-1})\Big].
\end{align*}
To prove the last term in the decomposition \eqref{eq:decomp_III} of order $o_p(1)$, we take the filtration ${\cal O}_{nj} = \sigma(\{O_1,\ldots,O_j\}, \bs{B}_{j+1})$ and
\begin{align*} 
  &H_{nj} := \frac{1}{\sqrt{n-q_n}}\frac{m_{\,j\,-1}}{\hat{\sigma}_{n\,j\,-1}}
  A_{j}\big[\tilde{Y}_{j}-E_0(1,B_{j,\,k_{j-1}},k_{j-1})\big]\\
  &\hspace{2.5cm} \times \bigg\{\frac{1}{\mathbb{Q}_{j-1}(0)}\big[\hat{Q}_{j-1}-Q\big](B_{j,\,k_{j-1}},k_{j-1}) + \bigg(\frac{1}{\mathbb{Q}_{j-1}(0)}-\frac{1}{Q(0)}\bigg)Q(B_{j,\,k_{j-1}},k_{j-1})\bigg\}
\end{align*}
with $E\big[\tilde{Y}_{j+1}\big|A_{j+1}=1,\mathcal{O}_{nj}\big] = E_0(1,B_{j+1,\,k_{j}},k_{j})$, and apply martingale difference sequence theory. Combining the above results, along with the conditions of Theorem \ref{Thm:stab_one_step},
yields that $\mbox{(III)}1_{\mathcal{L}_n} = o_p(1)$.
\end{proof}

\begin{lemma}\label{lemma:asymp.V}
Under the conditions of Theorem \ref{Thm:stab_one_step}, $\mbox{(V)}1_{\mathcal{L}_n}$ is asymptotically negligible.
\end{lemma}
\begin{proof}
  It is trivial to see that $\mbox{(V)}1_{{\cal L}_n}=o_p(1)$ under the null. Similar arguments for Lemma S6.22 of \cite{Huang2023} can be used to verify it under the alternative.
\end{proof}

\end{document}